%% file: main.tex
\documentclass[11pt]{article}
\usepackage{booktabs}
\usepackage{subcaption}

\input{preamble}
\input{macros}

\input{title}

\begin{document}

\maketitle

\input{sec0_intro}
\input{sec1}
\input{sec2}

\input{sec3}

\input{sec4}
\input{sec5_discussion}

\input{Acknowledgments}

\newpage
\appendix
\input{appendixA}
\input{appendixB}
\input{appendixC}

\bibliographystyle{JHEP}
\bibliography{ref}

\end{document}

%% file: preamble.tex
\usepackage{subcaption}

\usepackage{jheppub}

\usepackage{latexsym,cancel,amssymb}
\usepackage{graphicx}
\usepackage{epstopdf}   
\usepackage{epsfig}   
\usepackage{slashed, xcolor}
\usepackage{multirow}
\usepackage{verbatim}
\usepackage{booktabs}
\usepackage{hyperref}

\usepackage{amsfonts}
\usepackage{mathrsfs}
\usepackage{amsmath}
\usepackage{float}
\usepackage{framed}
\usepackage[medium]{titlesec}
\usepackage{bm}
\usepackage[normalem]{ulem}
\usepackage{extarrows}
\usepackage{isodateo}
\usepackage{indentfirst}
\usepackage{blindtext}
\usepackage[export]{adjustbox}
\usepackage{wrapfig}
\usepackage{mathtools}
\usepackage{enumerate}
\usepackage{mdframed} 
\usepackage{physics}
\usepackage{tikz}
\usetikzlibrary{intersections,calc}
\usepackage{amsthm}

%% file: macros.tex
\input{preamble}
\newcommand{\beq}{\begin{equation}}
\newcommand{\eeq}{\end{equation}}
\newcommand{\bea}{\begin{eqnarray}}
\newcommand{\eea}{\end{eqnarray}}

\renewcommand{\d}{{\mathrm{d}}}

\newcommand{\besub}{\begin{subequations}}
\newcommand{\eesub}{\end{subequations}}

\newcommand{\nnn}{\nonumber \\ }

\newcommand{\rmi}{\mathrm{i}}	

\newcommand{\h}{\bar{h}}
\newcommand{\LB}{\hat{L}_B}
\newcommand{\F}{{\cal F}}

\definecolor{darkerblue}{rgb}{0.2,0.2,0.5}
\definecolor{seagreen}{rgb}{0.180392,0.545098,0.341176}
\definecolor{smagenta}{rgb}{0.5,0.145098,0.341176}
\definecolor{deepblue}{rgb}{0,0,1}

\newtheorem{lemma}{Lemma}
\newtheorem{theorem}{Theorem}

\newtheorem{corollary}{Corollary}

%% file: title.tex
\title{Particle productions during collisions of highly boosted bubble walls}

\author[a,b]{Haipeng An}
\author[a]{Hongyi Jiang}
\affiliation[a]{Department of Physics, Tsinghua University, Beijing 100084, China}
\affiliation[b]{Center for High Energy Physics, Tsinghua University, Beijing 100084, China}

\emailAdd{anhp@mail.tsinghua.edu.cn}
\emailAdd{jianghy23@mails.tsinghua.edu.cn}

\abstract{

We investigate the production of particles much heavier than the characteristic scale of a cosmological first-order phase transition through collisions of highly boosted bubble walls. Using the scalar order-parameter field, we derive the ultraviolet behavior of its Fourier-space profile for both elastic and inelastic collisions. In the regime $\chi\equiv\omega^2-\mathbf{k}^2\gg M_h^2$, we find the universal result
$\tilde{\phi}(\chi) = -2V^\prime(2v_\phi)\chi^{-2}+O(\chi^{-3}),$
implying that the spectral density scales as $F(\chi)\propto [V^\prime(2v_\phi)]^2\chi^{-4}$. Thus, heavy-particle production is localized near the instant of collision and, at leading order, depends on the scalar potential only through $V^\prime(2v_\phi)$. We verify this behavior using high-precision numerical solutions of the trapping equation, carefully suppressing spectral leakage from the finite integration domain, and obtain agreement over a broad ultraviolet range. We then derive analytical production rates for general heavy-particle thresholds and for fermion pairs, together with their cosmological number density and yield. Finally, we extend the analysis to $(3+1)$ dimensions and incorporate the finite bubble radius, finding an order-one suppression relative to the parallel-wall approximation. Our results revise the ultraviolet scaling used in previous treatments and have direct implications for superheavy dark-matter production and baryogenesis from bubble collisions.

}

%% file: sec0_intro.tex
\newpage

\section{Introduction}

A cosmological first-order phase transition (FOPT) is among the most dramatic phenomena that may have occurred in the early Universe. The nucleation, expansion, and subsequent collision of vacuum bubbles constitute a highly dynamical, far-from-equilibrium process. Bubble collisions can generate a stochastic gravitational-wave background~\cite{Witten:1984rs}, while the accompanying departure from thermal equilibrium can provide the conditions required for successful baryogenesis or leptogenesis~\cite{Sakharov:1967dj}. Such collisions may also produce superheavy dark matter in the early Universe~\cite{An:2022toi,Freese:2023fcr,Cheng:2026npt}.

In scenarios involving runaway bubbles, the Lorentz factor $\gamma_w$ of the bubble walls at the time of collision can be much larger than unity. For phase transitions driven by the decay of a metastable vacuum, the boosted energy scale $m_\phi\gamma_w$ may even approach the Planck scale~\cite{Freese:2023fcr,Casey:2024jep,An:2026ghw,An:2026hiq}, where $m_\phi$ denotes the characteristic mass scale of the phase-transition sector. The enormous energy concentrated in these highly localized collisions can therefore produce particles far heavier than the intrinsic scale of the transition. Such particles are compelling candidates for superheavy dark matter~\cite{Falkowski:2012fb,Iso:2017uuu,Hambye:2018qjv,Baratella:2018pxi,Baldes:2021aph,Giudice:2024tcp,Ai:2024ikj,Shakya:2025qpi,Fairbairn:2026mck}. Alternatively, they may be heavy right-handed neutrinos whose subsequent out-of-equilibrium decays generate a lepton asymmetry, which is later converted into the baryon asymmetry of the Universe through electroweak sphaleron processes~\cite{Konstandin:2011ds,Katz:2016adq,Cataldi:2024pgt,Datta:2024tne,Shakya:2025qpi}.

The dynamics of the phase transition can be described in terms of the scalar order-parameter field $\phi$. As discussed in detail in Sec.~\ref{subsec:PP}, the production rate of particles with masses greater than $m$ is proportional to
\begin{equation}
    \int_{\omega^2-\mathbf{k}^2>m^2}
    \frac{\d^3\mathbf{k}\,\d\omega}{(2\pi)^4}\,
    \left|\tilde{\phi}(\mathbf{k},\omega)\right|^2
    \operatorname{Im}
    \left[
        \tilde{\Gamma}^{(2)}
        \left(\omega^2-\mathbf{k}^2\right)
    \right],
\label{eq:prod_rate}
\end{equation}
where $\tilde{\Gamma}^{(2)}$ is the Fourier transform of the two-point effective action. This quantity can typically be evaluated perturbatively using standard Feynman-diagram techniques. By contrast, the Fourier-transformed field configuration $\tilde{\phi}(\mathbf{k},\omega)$ encodes the nonperturbative bubble dynamics and generally cannot be obtained through perturbation theory. In particular, its high-energy, or equivalently high-frequency, tail controls the production of particles whose masses greatly exceed the characteristic energy scale of the phase transition. In this work, we determine this high-energy spectrum systematically using analytical approximations supplemented by high-precision numerical simulations.

Before the collision, each bubble wall propagates ultrarelativistically, with $\gamma_w\gg1$. In this regime, the local profile of an isolated bubble wall is well approximated by a function of the spacetime interval
\begin{align}
    s=\sqrt{t^2-\mathbf{x}^2}.
\end{align}
The corresponding local Lorentz invariance implies that an isolated bubble cannot efficiently produce particles whose masses greatly exceed the characteristic scale that determines the wall thickness in its rest frame. Equivalently, the Fourier transform of the single-bubble configuration has negligible support in the high-invariant-mass regime
\begin{align}
    \chi\coloneqq\omega^2-\mathbf{k}^2\gg M_h^2.
\end{align}

After the bubble walls collide, the field evolution falls into two qualitatively distinct regimes:
\begin{itemize}
    \item \textbf{Inelastic collisions:}
    The vacuum-energy difference is sufficiently large that, after the collision, the field in each energetic shell oscillates around the true vacuum.

    \item \textbf{Elastic collisions:}
    The vacuum-energy difference is relatively small, and the field oscillates around the false vacuum after the collision.
\end{itemize}
The post-collision profiles in these two regimes are shown in the third panels of the upper and lower rows of Fig.~\ref{fig:profile1}, respectively. For comparison, the second panels show the corresponding field configurations at the instant when the two bubble walls first come into contact.

In both regimes, the collision generates substantial high-frequency contributions to $\tilde{\phi}(\mathbf{k},\omega)$ at invariant masses
\begin{align}
    \omega^2-\mathbf{k}^2\gg m_\phi^2.
\end{align}
The highly boosted walls produce sharp, narrow peaks in the field profile, as shown in the middle panels of Fig.~\ref{fig:profile1}, whose widths are much smaller than the Compton wavelength $m_\phi^{-1}$. Similar short-distance structures persist in the post-collision oscillations. Accurately determining the high-frequency behavior of $\tilde{\phi}(\mathbf{k},\omega)$ therefore requires a detailed understanding of the field dynamics within these localized transition regions.

In this paper, we present a systematic analysis of the spectrum of $\tilde{\phi}$. We show that the spectral density
\begin{equation}
    F(\chi)
    =
    \int
    \frac{\d^3\mathbf{k}\,\d\omega}{(2\pi)^4}\,
    \left|\tilde{\phi}(\mathbf{k},\omega)\right|^2
    \delta\left(\omega^2-\mathbf{k}^2-\chi\right)
\label{eq:spectral_density}
\end{equation}
scales as $\chi^{-4}$ in the high-energy regime $\chi\gg m_\phi^2$ for both elastic and inelastic collisions. This scaling differs from previous results in the literature~\cite{Konstandin:2011ds,Mansour:2023fwj} and significantly modifies the predictions of baryogenesis and dark-matter scenarios based on collisions of highly boosted bubble walls.

While this work was being completed, Ref.~\cite{Ghoshal:2026pew} appeared, presenting a Feynman-like parton-model study of particle production in bubble collisions. In that work, the process is interpreted as collisions between particles associated with the two bubble walls, and heavy-particle production is evaluated perturbatively using approximate wall configurations as initial states. By contrast, our framework yields the analytical production rates given in Eqs.~\eqref{eq:NA_heavy}, \eqref{eq:NA_general_approx}, and \eqref{eq:NA_final}. These expressions apply, in principle, to a broad class of potentials when the mass of the produced particle is much larger than the characteristic energy scale of the phase-transition sector. In particular, we show that the leading heavy-particle production rate depends only on $V^\prime(2v_\phi)$, where $V^\prime\equiv\partial V/\partial\phi$ and $v_\phi$ is the vacuum expectation value of $\phi$ in the true vacuum. This result has a simple physical interpretation: superheavy particles are produced predominantly at the instant when the bubble walls collide, at which point the field temporarily reaches values near $\phi\simeq2v_\phi$. We further solve the trapping equation~\eqref{eq:trapping} numerically and directly compute the spectrum $\tilde{\phi}$, finding good agreement with our analytical predictions. Our results are also qualitatively consistent with the results of Ref.~\cite{Ghoshal:2026pew}.

The remainder of this paper is organized as follows. In Sec.~\ref{sec:analytical}, we introduce the phase-transition model and analyze bubble collisions in $(1+1)$ dimensions, deriving explicitly the $\chi^{-2}$ scaling of $\tilde{\phi}$ in the regime $\chi\gg m_\phi^2$. In Sec.~\ref{sec:numerical}, we formulate the $(1+1)$-dimensional lattice framework used to simulate two-bubble collisions, and show that the corresponding numerical results agree with the analytical result in Sec.~\ref{sec:analytical}. In Sec.~\ref{sec:production}, we calculate the production rates for different particle physics models. In Sec.~\ref{sec:particle}, we generalize our analysis to $(3+1)$-dimensional Minkowski spacetime and evaluate the associated geometric correction factors. Finally, in Sec.~\ref{sec:discussion}, we compare our formalism with previous treatments in the literature and discuss the implications of our results.

\begin{figure}[htbp]
    \centering
    \includegraphics[width=0.98\linewidth]{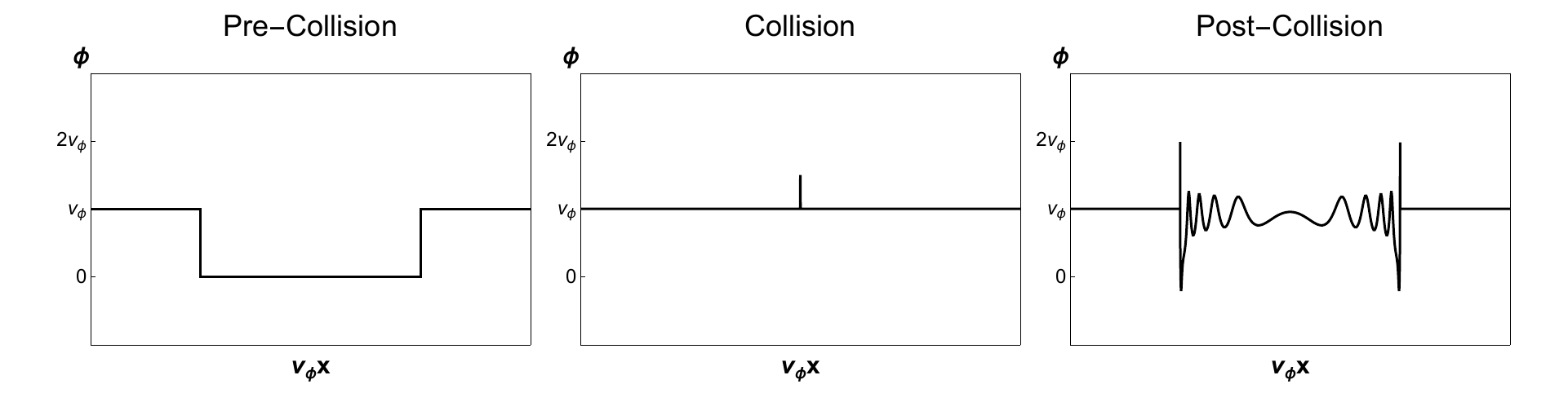}
    \includegraphics[width=0.98\linewidth]{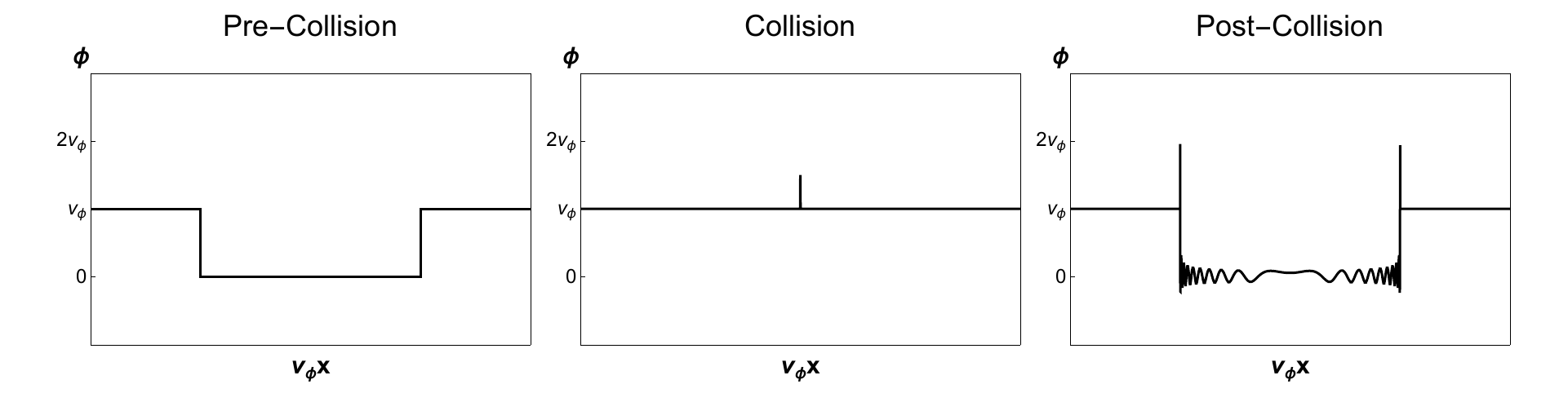}
    \caption{Field configurations during inelastic (upper row) and elastic (lower row) bubble-wall collisions. The panels illustrate the evolution from the approaching walls, through their initial contact, to the post-collision field dynamics.}
    \label{fig:profile1}
\end{figure}

%% file: sec1.tex
\section{Analytical Analysis of Particle Production Rate During Bubble Wall Collisions}
\label{sec:analytical}

\subsection{Generic Field Configuration for Two-Bubble Collisions in \texorpdfstring{$1+1$}{}D}

We consider a scalar field $\phi$ with a quartic potential $V(\phi)$ in $1+1$D Minkowski spacetime. The Lagrangian density is given by
\begin{equation}
\label{eq:Lag}
    \mathcal{L} = \frac{1}{2} (\partial_t \phi)^2 - \frac{1}{2} (\partial_x \phi)^2 - V(\phi) \, . 
\end{equation}
Since we are interested in theories exhibiting FOPT, we assume that the potential possesses two minima, located at $\phi = 0$ and $\phi = v_\phi$, separated by a potential barrier. At the classical level, an overall scaling factor in the potential can always be absorbed by a suitable redefinition of the field and coordinates. Following Ref.~\cite{Jinno:2019bxw}, the generic form of the polynomial potential with up to quartic interactions can be written as
\begin{equation}
\label{eq:potential}
    V(\phi) = a v_\phi^2 \phi^2 - (2 a + 4) v_\phi \phi^3 + (a+3) \phi^4 \, .
\end{equation}
The analytical analysis in this section is independent of the detailed form of the potential. Eq.~\eqref{eq:potential} is used in the numerical study in Sec.~\ref{sec:numerical}.

The evolution of highly boosted bubble walls can be described by the trapping equation~\cite{Jinno:2019bxw}. Inside the future light cone, the equation of motion derived from the Lagrangian~\eqref{eq:Lag} simplifies to
\begin{equation}
\label{eq:trapping}
    \partial_s^2 \phi + \frac{1}{s} \partial_s \phi = - V^\prime(\phi) \, , 
\end{equation}
where $s = \sqrt{t^2 - x^2}$ denotes the proper time of the bubble wall. The initial conditions are specified as $\phi(0) = 2 v_\phi$ and $\partial_s \phi(0) = 0$.

Let the solution to Eq.~\eqref{eq:trapping} be parametrized as $\phi(s) = v_\phi + v_\phi h(s)$, with $h(0)=1$ and $h^\prime(0)=0$ (as shown in Fig.~\ref{fig:profile1}). The full spacetime configuration of $\phi$ for the collisions of two bubbles is then expressed as
\begin{equation}
\label{eq:phitx}
    \phi(t,x) = v_\phi \left[ 1 - \theta(t^2-x^2) \theta(-t) + h(s) \theta(t^2 - x^2) \theta(t) \right] \, . 
\end{equation}

\subsection{Particle Production During Bubble Collisions in \texorpdfstring{$1+1$}{}D}
\label{subsec:PP}

The surface number density of particles produced by bubble collisions is given by~\cite{Watkins:1991zt,Konstandin:2011ds,Falkowski:2012fb,Mansour:2023fwj}:
\begin{equation}
\label{eq:particle_production}
    \frac{N}{A} = \frac{1}{2 \pi^2} \int \d \chi \, \mathrm{Im}\left[\tilde{\Gamma}^{(2)}(\chi)\right] f(\chi) \, , 
\end{equation}
where 
\begin{equation}
\label{eq:f}
    f(\chi) = \int \delta(\omega^2 - k^2 - \chi) |\tilde{\phi}(\omega,k)|^2 \, \mathrm{d}k \, \mathrm{d}\omega \, , 
\end{equation}
and $\tilde\phi(\omega, k)$ is the two-dimensional Fourier transform of $\phi(t,x)$.

Let $l_w$ denote the wall width in its rest frame, meaning that the wall thickness immediately prior to collision is $l_w/\gamma_w$. Consequently, we expect the Fourier transform $\tilde\phi$ to decay rapidly in the high-momentum/high-frequency region where $\omega, k > \gamma_w/l_w$. Introducing physical cutoffs into the momentum integrals in Eq.~\eqref{eq:f}, we obtain
\begin{align}
\label{eq:fchi}
    f(\chi) &= \int \theta(\gamma_w^2 / l_w^2 - \omega^2) \theta(\gamma_w^2 / l_w^2 - k^2) \delta(\omega^2 - k^2 - \chi) |\tilde{\phi}(\omega,k)|^2 \, \mathrm{d}k \, \mathrm{d}\omega \nonumber \\ 
    &= \int \theta(\gamma_w^2 / l_w^2 - \chi - k^2) \left|\tilde{\phi}\left(\sqrt{\chi + k^2},k\right)\right|^2 \, \frac{\mathrm{d}k}{\sqrt{k^2 + \chi}} \, ,
\end{align}
where we have exploited the reflection symmetry $|\tilde{\phi}(\omega,k)| = |\tilde{\phi}(-\omega,-k)|$.

Since the field configuration $\phi(t,x)$ in Eq.~\eqref{eq:phitx} is explicitly Lorentz invariant under proper Lorentz transformations, its Fourier transform $\tilde\phi$ depends on $\omega$ and $k$ solely through the invariant combination $\chi = \omega^2 - k^2$. Thus, $f(\chi)$ simplifies to
\begin{align}
\label{eq:f0}
    f(\chi) &= |\tilde{\phi}(\chi)|^2 \int \theta(\gamma_w^2 / l_w^2 - \chi - k^2) \, \frac{\mathrm{d}k}{\sqrt{k^2 + \chi}} \nonumber \\ 
    &= |\tilde{\phi}(\chi)|^2 \log\left[\frac{2 \gamma_w^2 / l_w^2 - \chi + 2 (\gamma_w / l_w) \sqrt{\gamma_w^2 / l_w^2 - \chi}}{\chi}\right] \theta(\gamma_w^2 / l_w^2 - \chi) \, . 
\end{align}
Because the two-point function $\tilde{\Gamma}^{(2)}(\chi)$ has support only on the physical upper light cone where $\omega^2 - k^2 > 0$, we restrict our analysis of $\tilde\phi(\chi)$ to the regime $\chi > 0$.

\subsection{Spectrum of \texorpdfstring{$\phi$}{}} 
\label{subsection:spectrum_phi}

To evaluate $f(\chi)$, we must compute the Fourier transform of the field configuration $\phi(t,x)$. We decompose the field as
\begin{equation}
    \phi(t,x) = v_\phi + \phi_1(t,x) + \phi_2(t,x) \, ,
\end{equation}
where
\begin{align}
    \phi_1(t,x) &= - v_\phi \theta(t^2-x^2) \theta(-t) \, , \\ 
    \phi_2(t,x) &= v_\phi h(s) \theta(t^2-x^2) \theta(t) \, .
\end{align}
The Fourier transform is then given by
\begin{equation}
    \tilde{\phi}(\omega,k) = \int \mathrm{d}^2 x \, \phi(t,x) \mathrm{e}^{\mathrm{i} \omega t - \mathrm{i} k x} 
    = (2 \pi)^2 v_\phi \delta(\omega) \delta(k) + \tilde{\phi}_1(\omega,k) + \tilde{\phi}_2(\omega,k) \, .
\end{equation}
The term $\tilde\phi_1(\omega,k)$ can be integrated directly to yield
\begin{equation}
\label{eq:phi1}
    \tilde{\phi}_1(\omega,k) = \frac{2 v_\phi}{(\omega - \mathrm{i} \epsilon)^2 - k^2} \, ,
\end{equation}
which coincides with the advanced propagator of a massless scalar field, reflecting the highly boosted initial state.

Similarly, we obtain the explicit expression for $\tilde\phi_2(\omega,k)$:
\begin{equation}
\label{eq:phi2}
    \tilde{\phi}_2(\omega,k) = v_\phi \int_0^\infty 2 s \, \mathrm{d}s \, h(s) \mathrm{K}_0 \left[ -\mathrm{i} s \sqrt{(\omega + \mathrm{i} \epsilon)^2 - k^2} \right] \, , 
\end{equation}
where $\mathrm{K}_0$ is the modified Bessel function of the second kind. The branch of the square root is chosen such that $\mathrm{Im} \left[ \sqrt{(\omega + \mathrm{i} \epsilon)^2 - k^2}\right] > 0$, ensuring the retarded nature of the propagation.

The trapping equation~\eqref{eq:trapping} implies that the profile function $h(s)$ satisfies
\begin{equation}
    \partial_s^2 h(s) + \frac{1}{s} \partial_s h(s) = - v_\phi^{-1} V^\prime(v_\phi + v_\phi h(s)) \, . 
\end{equation}
This indicates that $h(s)$ behaves as a damped oscillator, with friction inversely proportional to the proper time $s$, rendering the system asymptotically stable. Consequently, in the $s \to \infty$ limit, $V^\prime(v_\phi + v_\phi h(s))$ approaches $V^\prime(v_\infty) = 0$, where $v_\infty$ represents the vacuum expectation value after collision ($v_\infty = v_T$ for the inelastic case and $v_F$ for the elastic case). For the potential in Eq.~\eqref{eq:potential}, we have $v_T = v_\phi$ and $v_F = 0$.

This asymptotic stability allows us to evaluate the $s$-integral in Eq.~\eqref{eq:phi2} using a systematic asymptotic expansion. At leading order, since $V^\prime(v_\infty) = 0$, a linear expansion gives $V^\prime(\phi) \approx V^{\prime\prime}(v_\infty) (v_\phi + v_\phi h(s) - v_\infty)$. Defining the oscillation profile $\h(s)$ via
\begin{equation}
\label{eq:h0}
    v_\phi \h(s) \equiv v_\phi + v_\phi h(s) - v_\infty \, ,
\end{equation}
we can describe both elastic and inelastic collisions within a unified framework, where $v_\phi \h(s)$ parametrizes the late-time oscillations of $\phi$ around $v_\infty$. In the vicinity of $v_\infty$, the evolution of $\h$ is governed by
\begin{equation}
    \partial_s^2 \h(s) + \frac{1}{s}\partial_s \h(s) + M_h^2 \h(s) = 0 \, ,
\end{equation}
where the effective mass is $M_h^2 \equiv V^{\prime\prime}(v_\infty)$. For the potential~\eqref{eq:potential}, we have
\begin{equation}
    M_h^2 = 2 a v_\phi^2 - 12(a+2) v_\phi v_\infty + 12(a+3)v_\infty^2 \, .
\end{equation}
At large $s$, the amplitude of the envelope decays as $s^{-1/2}$, yielding the asymptotic behavior $\h(s) \sim s^{-1/2} \cos(M_h s + \varphi)$. We can decompose $V^\prime(\phi(s))$ as
\begin{equation}
\label{eq:Vp}
    V^\prime(\phi(s)) = V^{\prime\prime}(v_\infty) v_\phi \h(s) + v_\phi \Delta(\h(s)) \, ,
\end{equation}
where $v_\phi \Delta(\h(s))$ collects the nonlinear terms. For the potential~\eqref{eq:potential}, this is given by
\begin{equation}
    \Delta(\h) = 3 C_1 \h^2 + 4 C_2 \h^3 \, ,
\end{equation}
with the coefficients
\begin{align}
    C_1 &= -(2a+4) v_\phi^2 + 4(a+3) v_\phi v_\infty \, , \nonumber \\
    C_2 &= (3+a) v_\phi^2 \, .
\end{align}
Incorporating these nonlinear interactions, the full equation of motion for $\h$ becomes
\begin{equation}
\label{eq:h02}
    \left(\partial_s^2 + \frac{1}{s} \partial_s + M_h^2 \right) \h(s) = - \Delta(\h(s)) \, .
\end{equation}

Now, let us return to the evaluation of $\tilde\phi$. Note that the modified Bessel function $\mathrm{K}_0$ is an eigenfunction of the Bessel operator,
\begin{equation}
    \hat{L}_{\rm B} \equiv \partial_s^2 + \frac{1}{s}\partial_s \, ,
\end{equation}
satisfying
\begin{equation}
    \hat{L}_{\rm B} \mathrm{K}_0\left[ -\mathrm{i} s \sqrt{(\omega+\mathrm{i}\epsilon)^2-k^2} \right] = \left[-(\omega+\mathrm{i}\epsilon)^2+k^2\right] \mathrm{K}_0\left[ -\mathrm{i} s \sqrt{(\omega+\mathrm{i}\epsilon)^2-k^2} \right] \, .
\end{equation}
Consequently, we can rewrite $\tilde{\phi}_2(\omega,k)$ in Eq.~\eqref{eq:phi2} as
\begin{equation}
    \tilde{\phi}_2(\omega,k) = \int_0^\infty 2 s \, \mathrm{d}s \, v_\phi h(s) \frac{(\hat{L}_{\rm B} + M_h^2)}{M_h^2 - (\omega + \mathrm{i} \epsilon)^2 + k^2} \mathrm{K}_0 \left[ -\mathrm{i} s \sqrt{(\omega + \mathrm{i} \epsilon)^2 - k^2} \right] \, .
\end{equation}
Integrating by parts twice yields
\begin{align}
\label{eq:phi3}
    \tilde{\phi}_2(\omega,k) &= \int_0^\infty 2 s \, \mathrm{d}s \, v_\phi \left[\frac{(\hat{L}_{\rm B} + M_h^2)h(s)}{M_h^2 - (\omega + \mathrm{i} \epsilon)^2 + k^2}\right] \mathrm{K}_0 \left[ -\mathrm{i} s \sqrt{(\omega + \mathrm{i} \epsilon)^2 - k^2} \right] \nonumber \\
    &\quad + \frac{2 v_\phi \left[ s \partial_s h(s) \mathrm{K}_0(-\mathrm{i}s\sqrt{(\omega+\mathrm{i}\epsilon)^2 - k^2}) - s h(s) \partial_s \mathrm{K}_0(-\mathrm{i} s \sqrt{(\omega+\mathrm{i}\epsilon)^2 - k^2})\right]_0^\infty}{M_h^2 - (\omega+\mathrm{i}\epsilon)^2 + k^2} \, .
\end{align}
Using the small-argument limit $\partial_s \mathrm{K}_0(-\mathrm{i}s\sqrt{(\omega+\mathrm{i}\epsilon)^2 - k^2}) \to -s^{-1}$ as $s \to 0$, along with the boundary conditions $h(0)=1$ and $h^{\prime}(0)=0$, the boundary term at $s=0$ simplifies to $2 v_\phi$. 

From the definition of $\h(s)$ in Eq.~\eqref{eq:h0} and its equation of motion~\eqref{eq:h02}, we find
\begin{equation}
    v_\phi(\hat{L}_{\rm B} + M_h^2) h(s) = M_h^2(v_\infty - v_\phi) - v_\phi \Delta(\h(s)) \, .
\end{equation}
Using the identity $\int_0^\infty s \, \mathrm{d}s \, \mathrm{K}_0(-\mathrm{i} s \sqrt{\chi}) = -\chi^{-1}$ (which holds when $\chi$ has a positive imaginary part), we arrive at
\begin{align}
    \tilde\phi_2(\omega,k) &= \frac{2 v_\phi}{M_h^2 - (\omega+\mathrm{i}\epsilon)^2 + k^2} + \frac{2 M_h^2(v_\phi - v_\infty)}{\left[M_h^2-(\omega+\mathrm{i}\epsilon)^2 + k^2\right]\left[(\omega+\mathrm{i}\epsilon)^2-k^2\right]} \nonumber \\
    &\quad - \int_0^\infty 2s \, \mathrm{d}s \, \frac{v_\phi \Delta(\h(s))}{M_h^2 - (\omega+\mathrm{i}\epsilon)^2 +k^2} \mathrm{K}_0\left[ -\mathrm{i} s\sqrt{(\omega+\mathrm{i}\epsilon)^2-k^2} \right]\, .
\end{align}
Combining this with $\tilde\phi_1$ from Eq.~\eqref{eq:phi1}, we obtain the full spectrum:
\begin{align}
\label{eq:phiT}
    \tilde\phi(\omega,k) &= (2\pi)^2 v_\phi \delta(\omega) \delta(k) + v_\phi \frac{8 \mathrm{i} \epsilon \omega}{\left[(\omega-\mathrm{i}\epsilon)^2 - k^2 \right]\left[(\omega+\mathrm{i}\epsilon)^2 - k^2\right]}
    \nonumber \\
    &\quad - \frac{2 M_h^2 (2 v_\phi - v_\infty)}{\left[(\omega+\mathrm{i}\epsilon)^2 - k^2 - M_h^2\right]\left[(\omega+\mathrm{i}\epsilon)^2 - k^2\right]} + \frac{v_\phi I_0(\omega,k)}{(\omega+\mathrm{i}\epsilon)^2 - k^2 - M_h^2} \, ,
\end{align}
where
\begin{equation}
\label{eq:I}
    I_0(\omega,k) = \int_0^\infty 2 s \, \mathrm{d}s \, \Delta(\h(s)) \mathrm{K}_0\left[ -\mathrm{i} s \sqrt{(\omega+\mathrm{i}\epsilon)^2-k^2} \right] \, .
\end{equation}
Eq.~\eqref{eq:phiT} is valid for both elastic and inelastic collisions, reducing to the expression in Ref.~\cite{Falkowski:2012fb} in the limit $\Delta = 0$ and $\gamma_w \to \infty$.

The nonlinear interactions contribute to $\tilde\phi$ through the integral $I_0(\omega, k)$. Since $\Delta(\h)$ starts at order $\h^2$, it is an oscillating function of $s$ whose amplitude decays at least as $s^{-1}$. The modified Bessel function $\mathrm{K}_0$ decays asymptotically as $s^{-1/2}$. Consequently, the $s$-integral is convergent, except at resonances where the oscillation frequency of $\mathrm{K}_0(-\mathrm{i} s \sqrt{\chi})$ matches the frequencies present in $\Delta(\h(s))$. Because $\h(s)$ oscillates with a fundamental frequency $M_h$ as $s \to \infty$, resonances are expected at $\sqrt{\chi} = n M_h$ for integer $n$. For $n \geq 4$, the decay of $\Delta(\h(s)) \sim \mathcal{O}(s^{-n/2})$ is rapid enough to render the $s$-integral absolutely convergent. Thus, the physical resonances are dominated by the low-order terms ($\h^2$ and $\h^3$) in the polynomial expansion of $\Delta(\h)$. The $\h^2$ term drives resonances at $n=0$ and $n=2$ (producing poles at $\chi = 0$ and $\chi = 4 M_h^2$), while the $\h^3$ term drives resonances at $n=1$ and $n=3$ (producing poles at $\chi = M_h^2$ and $\chi = 9 M_h^2$). In an expanding universe, these poles are regularized by the Hubble parameter, and in $3+1$-dimensional spacetime, they are regularized by the finite bubble radius.

\subsection{Particle Production Rate in the Ultra-High Frequency Regime (\texorpdfstring{$\chi \gg M_h^2$}{})} 
\label{subsection:highfreq}

From Eq.~\eqref{eq:phiT}, we observe that when nonlinear interactions are neglected, $\tilde\phi$ scales as $\chi^{-2}$ in the deep ultraviolet (UV) regime ($\chi \gg M_h^2$). To assess the impact of nonlinearities, we examine the behavior of $I_0(\omega,k)$ in the same limit. Because the spectrum is free of poles in this regime, we can safely drop the $\mathrm{i}\epsilon$ regulator in the Bessel function argument, meaning $I_0$ depends solely on $\chi$:
\begin{equation}
    I_0(\chi) = \int_0^\infty 2 s \, \mathrm{d}s \, \Delta(\h(s)) \mathrm{K}_0(-\mathrm{i} s \sqrt{\chi}) \, .
\end{equation}

By utilizing the relation $\hat{L}_{\rm B} \mathrm{K}_0(-\mathrm{i} s \sqrt{\chi}) = -\chi \mathrm{K}_0(-\mathrm{i} s \sqrt{\chi})$ and performing integration by parts, we find
\begin{equation}
\label{eq:I_parts}
    I_0(\chi) = - \frac{2\Delta(\h(0))}{\chi} - \frac{1}{\chi} I_1(\chi) \, ,
\end{equation}
where
\begin{equation}
    I_1(\chi) = \int_0^\infty 2 s \, \mathrm{d}s \, \mathrm{K}_0(-\mathrm{i} s \sqrt{\chi}) \hat{L}_{\rm B} \Delta(\h(s)) \, .
\end{equation}
Since $\h(s)$ oscillates with frequency $M_h$ and an envelope of $\mathcal{O}(s^{-1/2})$, the action of the Bessel operator $\hat{L}_{\rm B}$ on $\Delta(\h(s))$ yields terms scaling as $s^{-2}$, $s^{-1} M_h$, or $M_h^2$. The first two cases introduce higher inverse powers of $s$, rendering the integrand absolutely integrable. By the Riemann-Lebesgue lemma, their contributions to $I_1$ are suppressed by negative powers of $\chi$ at high frequencies. The third case introduces a factor of $M_h^2$ without changing the large-$s$ scaling of the integrand. Therefore, by simple power counting, we expect $I_1(\chi)$ to scale identically to $I_0(\chi)$ as $\chi \to \infty$. This implies that a consistent power-counting scheme requires both $I_0(\chi)$ and $I_1(\chi)$ to scale as $\chi^{-1}$ in the large-$\chi$ limit (see Appendix~\ref{appendixB} for a mathematically rigorous proof of this scaling).

By repeatedly applying integration by parts, we generate the systematic asymptotic expansion:
\begin{equation}
    I_n(\chi) = - \frac{2 \hat{L}_{\rm B}^n \Delta(\h(0))}{\chi} - \frac{1}{\chi} I_{n+1}(\chi) \, ,
\end{equation}
which leads to the formal series for $I_0(\chi)$:
\begin{equation}
\label{eq:Iexpand}
    I_0(\chi) = - \frac{2\Delta(\h(0))}{\chi} + \frac{2\hat{L}_{\rm B}\Delta(\h(0))}{\chi^2} - \frac{2\hat{L}_{\rm B}^2\Delta(\h(0))}{\chi^3} + \frac{2\hat{L}_{\rm B}^3\Delta(\h(0))}{\chi^4} - \cdots
\end{equation}

Combining this result with the non-interacting contributions in Eq.~\eqref{eq:phiT}, the UV limit of the spectrum is dominated by:
\begin{equation}
    \tilde\phi(\chi) \approx \frac{1}{\chi^2} \left[ 2 M_h^2(v_\infty - 2 v_\phi) - 2\Delta(\h(0)) \right] \, .
\end{equation}
Using the relations in Eqs.~\eqref{eq:Vp} and \eqref{eq:h0}, this simplifies beautifully to:
\begin{equation}\label{eq:UVlim}
    \tilde\phi(\chi) = - \frac{2 V^\prime(\phi(0))}{\chi^2} + \mathcal{O}(\chi^{-3}) \, ,
\end{equation}
where $\phi(0) = 2 v_\phi$ is the initial field value at the moment of collision. This result is independent of the specific choice of the potential $V(\phi)$, requiring only that the potential admits a valid polynomial expansion.

Finally, substituting this back into Eq.~\eqref{eq:f0}, we obtain the high-frequency particle production spectrum in the regime $M_h^2 \ll \chi \ll \gamma_w^2 / l_w^2$:
\begin{equation}
    f(\chi) \to \frac{4 \left(V^\prime(\phi(0))\right)^2}{\chi^4} \theta(\gamma_w^2 / l_w^2 - \chi) \log{\left[ \frac{2 \gamma_w^2 / l_w^2 - \chi + 2 (\gamma_w / l_w) \sqrt{\gamma_w^2 / l_w^2 - \chi}}{\chi} \right]} \, . 
\end{equation}
This robust result holds universally for both elastic and inelastic bubble wall collisions.

While the derivation of the asymptotic expansion~\eqref{eq:Iexpand} presented in this section is heuristic because the intermediate integrals $I_n(\chi)$ are not absolutely integrable, a mathematically rigorous proof based on the Riemann-Lebesgue lemma is provided in Appendix~\ref{appendixB}.

%% file: sec2.tex
\section{Numerical Simulation of Bubble Collisions}
\label{sec:numerical}

In this section, we numerically solve the trapping equation,
Eq.~\eqref{eq:h02}, for the quartic potential in
Eq.~\eqref{eq:potential}. We then substitute the resulting field
profile into Eq.~\eqref{eq:I} and evaluate the Fourier-space profile
$\widetilde{\phi}(\chi)$ numerically. Several subtleties must be
addressed to obtain physically meaningful and numerically reliable
results. These include the truncation of the $s$-integration domain,
spectral leakage, the resolution of rapidly oscillating modes, and the
numerical treatment of resonant and infrared structures.

\subsection{Truncation of the \texorpdfstring{$s$}{s}-integral and the
\texorpdfstring{$\mathrm{i}\epsilon$}{i-epsilon} prescription}
\label{subsec:iepsilon}

A numerical calculation cannot evolve the field profile over the full
spacetime domain. We therefore solve for $\bar h(s)$ on a finite
interval,
\begin{equation}
    0\leq s\leq L,
\end{equation}
and truncate the integral in Eq.~\eqref{eq:I} at $s=L$. A sharp
truncation introduces an artificial discontinuity at the endpoint and
can contaminate the Fourier spectrum, particularly in the ultraviolet
(UV).

To regulate this endpoint contribution, we restore the
$\mathrm{i}\epsilon$ prescription. In the present context, this is
equivalent to multiplying the integrand by a smooth exponential window
that suppresses the contribution from the artificial boundary at
$s=L$. An arbitrary or naively
chosen window function can generate substantial spectral leakage. Such
leakage can dominate the true UV tail and lead to an incorrect
asymptotic scaling of $\widetilde{\phi}(\chi)$.


Fourier-type observables in simulations of cosmological phase
transitions are necessarily evaluated over a finite spatial or
temporal domain. Truncating an integral at a finite endpoint is
equivalent to multiplying the underlying source by a window function.
A sharp rectangular window produces spurious UV oscillations and
power-law tails, a standard phenomenon known as
\textit{spectral leakage}
\cite{harris1978use,press2007numerical}. If the physical spectrum
decreases more rapidly than the leakage term, the latter eventually
dominates and obscures the true asymptotic behavior.

In the present problem, the analytical result derived in
Sec.~\ref{sec:analytical} predicts
\begin{equation}
    \widetilde{\phi}(\chi)\propto\chi^{-2}
    \qquad
    (\chi\gg M_h^2).
    \label{eq:physical_UV_scaling}
\end{equation}
A hard cutoff at $s=L$, however, produces a spurious contribution
proportional to $\chi^{-1}$, which eventually dominates the physical
$\chi^{-2}$ tail.

This effect can be illustrated by a one-dimensional example. Consider
\begin{equation}
    f(x)=\frac{1}{2a}e^{-a|x|},
    \qquad
    \widetilde f(k)=\frac{1}{k^2+a^2}.
    \label{eq:leakage_example_f}
\end{equation}
Introduce the rectangular window
\begin{equation}
    W_L(x)=
    \begin{cases}
        1, & |x|<L,\\
        0, & |x|>L,
    \end{cases}
\end{equation}
whose Fourier transform is
\begin{equation}
    \widetilde W_L(k)=\frac{2\sin(kL)}{k}.
\end{equation}
The transform of the truncated function is the convolution
\begin{equation}
    \widetilde{fW_L}(k)
    =
    \int\frac{\dd p}{2\pi}\,
    \widetilde f(k-p)\widetilde W_L(p).
    \label{eq:convolution}
\end{equation}
Evaluating the integral gives
\begin{equation}
    \widetilde{fW_L}(k)
    =
    \frac{1}{k^2+a^2}
    \left[
        1-e^{-aL}
        \left(
            \cos(kL)-\frac{k}{a}\sin(kL)
        \right)
    \right].
    \label{eq:fW}
\end{equation}
When
\begin{equation}
    \frac{|k|}{a}e^{-aL}\gtrsim1,
\end{equation}
the second term in Eq.~\eqref{eq:fW} yields a contribution that scales
as $k^{-1}$ rather than the physical $k^{-2}$ behavior. This is a
simple manifestation of spectral leakage caused by the artificial
boundary.

The same mechanism operates in the present calculation. The physical
$\chi^{-2}$ behavior originates from the nonanalyticity of the
collision profile at $s=0$, specifically the discontinuity in its
derivative. By contrast, a hard cutoff introduces an additional,
unphysical discontinuity at $s=L$, producing a $\chi^{-1}$ tail.
Since particle production is associated with the collision of the two
bubble walls, as illustrated in Fig.~\ref{fig:profile1}, any dominant
contribution associated with the arbitrary numerical endpoint must be
discarded as unphysical.

The $\mathrm{i}\epsilon$ prescription replaces the sharp cutoff by an
exponentially decaying window. When the damping exponent at the
boundary is large, as ensured by Eq.~\eqref{eq:N3}, the artificial
endpoint contribution is exponentially suppressed and the physical
UV scaling can be recovered.

\subsection{Finite grid resolution and UV computational cost}
\label{subsec:UVcost}

To resolve oscillations with characteristic frequency $\sqrt{\chi}$,
the grid spacing $d$ must be sufficiently small. We impose the
resolution criterion
\begin{equation}
    \frac{\pi}{d}=N_1\sqrt{\chi},
    \label{eq:N1}
\end{equation}
where $N_1\gg1$ is a numerical safety factor.

Under the $\mathrm{i}\epsilon$ prescription, the frequency is shifted
according to
\begin{equation}
    \sqrt{\chi}
    =\sqrt{\omega^2-k^2}
    \;\longrightarrow\;
    \sqrt{(\omega+\mathrm{i}\epsilon)^2-k^2}
    \simeq
    \sqrt{\chi}
    +\mathrm{i}\frac{\omega\epsilon}{\sqrt{\chi}},
    \label{eq:complex_frequency}
\end{equation}
where terms of order $\epsilon^2$ have been neglected. To prevent the
regulator from significantly modifying the physical spectrum, the
imaginary part must remain small compared with the real part. We
parameterize this requirement as
\begin{equation}
    \omega\epsilon=\frac{\chi}{N_2},
    \qquad N_2\gg1.
    \label{eq:N2}
\end{equation}

At the same time, the endpoint contribution at $s=L$ must be
exponentially suppressed. We therefore require
\begin{equation}
    \left(\frac{\omega\epsilon}{\sqrt{\chi}}\right)L=N_3,
    \qquad N_3\gg1.
    \label{eq:N3}
\end{equation}
Combining Eqs.~\eqref{eq:N1}--\eqref{eq:N3}, the required number of
grid intervals is
\begin{equation}
    N=\frac{L}{d}
    =\frac{N_1N_2N_3}{\pi}.
    \label{eq:Ntotal}
\end{equation}
Thus, in the deep UV and at fixed relative precision, the required
grid size is approximately independent of the physical scale $\chi$.
Increasing $\sqrt{\chi}$ requires a smaller grid spacing, but the
corresponding integration interval decreases by the same factor.

\subsection{Resonant and infrared regimes}
\label{subsec:resonant_IR}

As discussed in Sec.~\ref{subsection:spectrum_phi},
$\widetilde{\phi}(\chi)$ contains physical resonances near
\begin{equation}
    \chi\simeq M_h^2,\qquad
    4M_h^2,\qquad
    9M_h^2.
\end{equation}
These narrow structures are difficult to resolve using a finite
$\mathrm{i}\epsilon$ regulator. The regulator damps the long-range
coherent oscillations responsible for the resonances and therefore
artificially broadens and suppresses the corresponding peaks. A
different numerical treatment is required to reconstruct their widths
and amplitudes reliably.

In the infrared (IR), $\chi<M_h^2$, the grid must resolve the intrinsic
mass scale $M_h$, rather than only the external scale $\sqrt{\chi}$.
The resolution criterion becomes
\begin{equation}
    \frac{\pi}{d}=N_1M_h.
    \label{eq:IRresolution}
\end{equation}
Using the same regulator conditions as above, the required number of
grid intervals scales as
\begin{equation}
    N=\frac{L}{d}
    =\frac{N_1N_2N_3}{\pi}
      \frac{M_h}{\sqrt{\chi}}.
    \label{eq:NIR}
\end{equation}
The computational cost therefore grows as $\chi^{-1/2}$ in the deep
IR.

A quantitatively reliable treatment of the resonant and IR regions
requires a hybrid method capable of isolating and analytically
resumming the long-distance oscillatory contributions. Developing
such a method would require additional mathematical machinery and is
beyond the scope of the present work. We therefore restrict the
numerical analysis below to the UV regime,
\begin{equation}
    \chi\gg M_h^2,
\end{equation}
and defer the detailed IR and resonant spectra to future work.

\subsection{Numerical evaluation of
\texorpdfstring{$\widetilde{\phi}(\chi)$}{phi(chi)} in the UV}
\label{subsec:numerical_UV}

\subsubsection{Simulation setup}
\label{subsubsec:simulation_setup}

For all UV simulations, we use
\begin{equation}    N=2^{24}\simeq1.68\times10^7
\end{equation}
grid intervals. In evaluating $\widetilde{\phi}$, we choose
\begin{equation}
    \epsilon=\frac{100}{L}.
    \label{eq:epsilon_choice}
\end{equation}
Equation~\eqref{eq:N3} then gives
\begin{equation}
    N_3=100\frac{\omega}{\sqrt{\chi}}\geq100,
\end{equation}
since $\omega/\sqrt{\chi}\geq1$. The contribution from the endpoint
$s=L$ is therefore exponentially suppressed.

The mass $M_h$ provides the intrinsic unit of the problem. We perform
three simulations with grid spacings
\begin{equation}
    d=
    2.5\times10^{-4}M_h^{-1},\qquad
    2.5\times10^{-5}M_h^{-1},\qquad
    5.0\times10^{-6}M_h^{-1}.
    \label{eq:grid_spacings}
\end{equation}
Because $L=Nd$, the corresponding integration lengths are
\begin{equation}
    L\simeq
    4.19\times10^3M_h^{-1},\qquad
    4.19\times10^2M_h^{-1},\qquad
    83.9M_h^{-1}.
    \label{eq:integration_lengths}
\end{equation}
Requiring the total numerical uncertainty to remain below approximately
$1\%$, these runs cover the respective ranges
\begin{align}
    6M_h &<\sqrt{\chi}<40M_h, \nonumber\\
    24M_h&<\sqrt{\chi}<400M_h, \nonumber\\
    120M_h&<\sqrt{\chi}<2000M_h.
    \label{eq:chi_ranges}
\end{align}
The overlap among these intervals provides a direct convergence test.


\subsection{Numerical results}
\label{subsec:UV_results}

\begin{table}[htbp]
    \centering
    \caption{Ranges of $\sqrt{\chi}/v_\phi$ covered by the three
    numerical grids in the inelastic and elastic collision regimes.}
    \label{tab:UV_ranges}

    \textbf{Inelastic collisions}\\[0.4em]
    \resizebox{\textwidth}{!}{%
    \begin{tabular}{c|ccc}
        \toprule
        & $6M_h<\sqrt{\chi}<40M_h$
        & $24M_h<\sqrt{\chi}<400M_h$
        & $120M_h<\sqrt{\chi}<2000M_h$\\
        \midrule
        $a=1.0$
        & $2.2\times10^{1}<\sqrt{\chi}/v_\phi<1.5\times10^{2}$
        & $9.0\times10^{1}<\sqrt{\chi}/v_\phi<1.5\times10^{3}$
        & $4.5\times10^{2}<\sqrt{\chi}/v_\phi<7.5\times10^{3}$\\
        $a=3.0$
        & $2.6\times10^{1}<\sqrt{\chi}/v_\phi<1.7\times10^{2}$
        & $1.0\times10^{2}<\sqrt{\chi}/v_\phi<1.7\times10^{3}$
        & $5.1\times10^{2}<\sqrt{\chi}/v_\phi<8.5\times10^{3}$\\
        $a=5.0$
        & $2.8\times10^{1}<\sqrt{\chi}/v_\phi<1.9\times10^{2}$
        & $1.1\times10^{2}<\sqrt{\chi}/v_\phi<1.9\times10^{3}$
        & $5.6\times10^{2}<\sqrt{\chi}/v_\phi<9.4\times10^{3}$\\
        $a=7.0$
        & $3.1\times10^{1}<\sqrt{\chi}/v_\phi<2.0\times10^{2}$
        & $1.2\times10^{2}<\sqrt{\chi}/v_\phi<2.0\times10^{3}$
        & $6.1\times10^{2}<\sqrt{\chi}/v_\phi<1.0\times10^{4}$\\
        \bottomrule
    \end{tabular}}

    \vspace{0.8em}
    \textbf{Elastic collisions}\\[0.4em]
    \resizebox{\textwidth}{!}{%
    \begin{tabular}{c|ccc}
        \toprule
        & $6M_h<\sqrt{\chi}<40M_h$
        & $24M_h<\sqrt{\chi}<400M_h$
        & $120M_h<\sqrt{\chi}<2000M_h$\\
        \midrule
        $a=13.0$
        & $3.1\times10^{1}<\sqrt{\chi}/v_\phi<2.0\times10^{2}$
        & $1.2\times10^{2}<\sqrt{\chi}/v_\phi<2.0\times10^{3}$
        & $6.1\times10^{2}<\sqrt{\chi}/v_\phi<1.0\times10^{4}$\\
        $a=17.0$
        & $3.5\times10^{1}<\sqrt{\chi}/v_\phi<2.3\times10^{2}$
        & $1.4\times10^{2}<\sqrt{\chi}/v_\phi<2.3\times10^{3}$
        & $7.0\times10^{2}<\sqrt{\chi}/v_\phi<1.2\times10^{4}$\\
        $a=21.0$
        & $3.9\times10^{1}<\sqrt{\chi}/v_\phi<2.6\times10^{2}$
        & $1.6\times10^{2}<\sqrt{\chi}/v_\phi<2.6\times10^{3}$
        & $7.8\times10^{2}<\sqrt{\chi}/v_\phi<1.3\times10^{4}$\\
        $a=25.0$
        & $4.2\times10^{1}<\sqrt{\chi}/v_\phi<2.8\times10^{2}$
        & $1.7\times10^{2}<\sqrt{\chi}/v_\phi<2.8\times10^{3}$
        & $8.5\times10^{2}<\sqrt{\chi}/v_\phi<1.4\times10^{4}$\\
        \bottomrule
    \end{tabular}}
\end{table}

\begin{figure}[htbp]
    \centering
    \begin{subfigure}{0.96\textwidth}
        \centering
        \includegraphics[width=\textwidth]{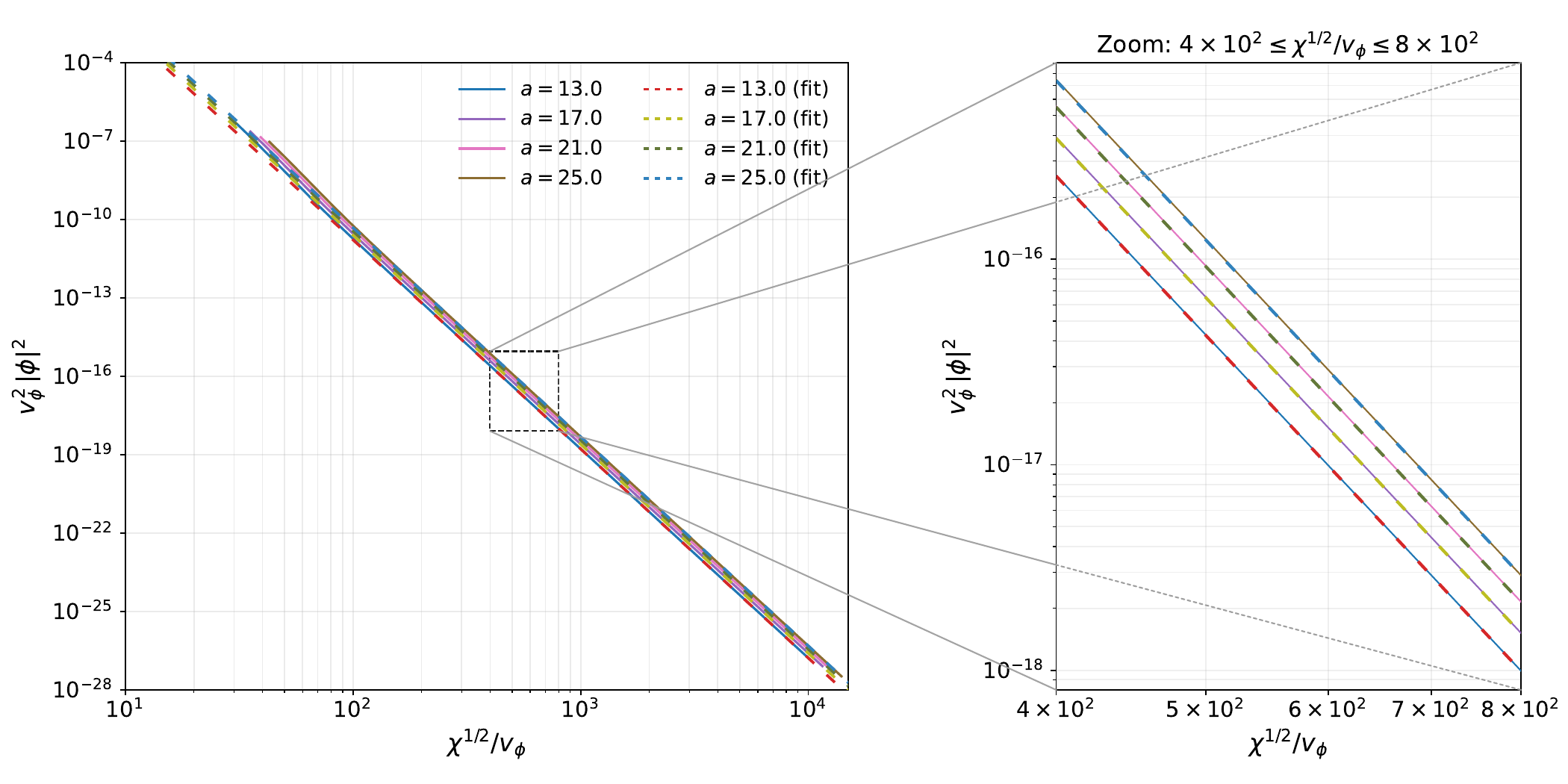}
        \caption{Inelastic collisions.}
        \label{fig:UV_inelastic}
    \end{subfigure}

    \vspace{0.5em}

    \begin{subfigure}{0.96\textwidth}
        \centering
        \includegraphics[width=\textwidth]{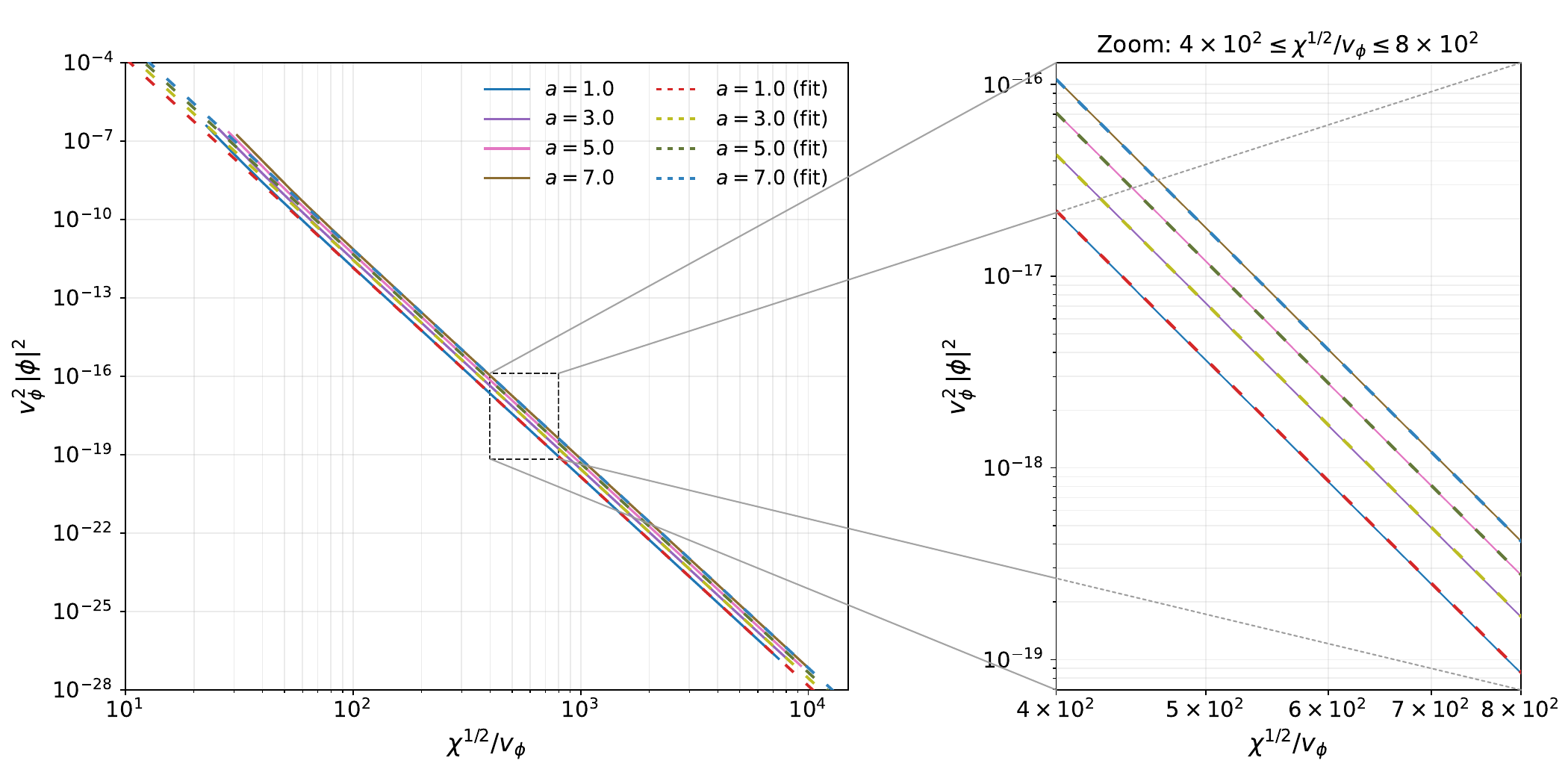}
        \caption{Elastic collisions.}
        \label{fig:UV_elastic}
    \end{subfigure}

    \caption{UV behavior of $\widetilde{\phi}(\chi)$ for
    inelastic (upper panel) and elastic (lower panel) bubble
    collisions. Solid curves show the numerical results for different
    values of the potential parameter $a$, while dashed curves show
    the analytical prediction in Eq.~\eqref{eq:UVlim}. In both cases,
    the numerical results reproduce the expected
    $\widetilde{\phi}\propto\chi^{-2}$ scaling.}
    \label{fig:UV}
\end{figure}

Figure~\ref{fig:UV} presents the numerical results for
$\widetilde{\phi}(\chi)$ in the inelastic and elastic collision
regimes. The solid curves show the numerical results for several
values of the potential parameter $a$ in
Eq.~\eqref{eq:potential}, while the dashed curves show the
corresponding analytical UV predictions from Eq.~\eqref{eq:UVlim}.

For both collision regimes, the numerical spectra agree well with the
analytical prediction when $\chi\gg M_h^2$ and exhibit the expected
scaling
\begin{equation}
    \widetilde{\phi}(\chi)\propto\chi^{-2}.
\end{equation}
The results obtained from the three numerical grids overlap smoothly
in their common domains. This agreement provides a nontrivial check of
the discretization, the treatment of the endpoint regulator, and the
absence of appreciable spectral leakage in the displayed range.

For completeness, Table~\ref{tab:UV_ranges} lists the corresponding
ranges of $\sqrt{\chi}/v_\phi$ for each value of $a$. The agreement
across the three independently simulated intervals further verifies
the $\chi^{-2}$ UV behavior for both elastic and inelastic bubble
collisions.

%% file: sec3.tex
\section{Heavy-Particle Production (\texorpdfstring{$m\gg M_h$}{m >> Mh})}
\label{sec:production}

When the mass scale of the produced particles is much larger than
$M_h$, particle production is controlled by the high-frequency
behavior of $f(\chi)$ derived in Eq.~\eqref{eq:UVlim}. Substituting
this asymptotic expression into Eq.~\eqref{eq:particle_production}
gives
\begin{align}
    \frac{N}{A}
    =
    \frac{4\bigl[V'(2v_\phi)\bigr]^2}{2\pi^2}
    \int_0^{\gamma_w^2/l_w^2}\!\dd\chi\,
    \frac{\operatorname{Im}\widetilde{\Gamma}^{(2)}(\chi)}
    {\chi^2(\chi-M_h^2)^2}
    \log\left[
        \frac{
            2\gamma_w^2/l_w^2-\chi
            +2(\gamma_w/l_w)
             \sqrt{\gamma_w^2/l_w^2-\chi}
        }{\chi}
    \right].
    \label{eq:NA_heavy}
\end{align}
The corresponding energy deposited per unit collision area is
\begin{equation}
    \frac{E}{A}
    =
    \int_0^{\gamma_w^2/l_w^2}\!\dd\chi\,
    \sqrt{\chi}\,
    \frac{\dd [N(\chi)/A]}{\dd\chi}.
    \label{eq:EA_heavy}
\end{equation}

The particle-number and energy integrals are dominated by their lower
kinematic limits provided that
$\operatorname{Im}\widetilde{\Gamma}^{(2)}(\chi)$ grows more slowly
than $\chi^3$ and $\chi^{5/2}$, respectively. For renormalizable
interactions, its asymptotic growth is at most linear in $\chi$.
Consequently, the observables considered below are insensitive to the
extreme-UV region, apart from the logarithmic dependence on the
kinematic cutoff.

\subsection{General threshold parametrization}
\label{subsec:general_threshold}

We parameterize the absorptive part of the two-point function as
\begin{equation}
    \operatorname{Im}\widetilde{\Gamma}^{(2)}(\chi)
    =
    \Gamma_0\,
    \chi^{\alpha_\Gamma}
    \left(1-\frac{M^2}{\chi}\right)^{\beta_\Gamma}
    \Theta(\chi-M^2),
    \label{eq:ImGamma_parametrization}
\end{equation}
where $\Gamma_0$ contains the relevant couplings, $M$ is the threshold
invariant mass, and $\alpha_\Gamma$ and $\beta_\Gamma$ are determined
by the interaction and the spins of the final-state particles.

For $M>M_h$, define
\begin{equation}
    \Lambda\equiv\frac{\gamma_w}{Ml_w}\gg1
\end{equation}
and rescale $\chi\rightarrow M^2\chi$. Equation~\eqref{eq:NA_heavy}
then becomes
\begin{align}
    \frac{N}{A}
    =
    \frac{4\bigl[V'(2v_\phi)\bigr]^2\Gamma_0}
         {2\pi^2M^{6-2\alpha_\Gamma}}
    \int_1^{\Lambda^2}\!\dd\chi\,
    \frac{
        \left(1-\chi^{-1}\right)^{\beta_\Gamma}
    }{
        \chi^{4-\alpha_\Gamma}
        \left[
            1-(M_h/M)^2\chi^{-1}
        \right]^2
    }
    \log\left[
        \frac{
            2\Lambda^2-\chi
            +2\Lambda\sqrt{\Lambda^2-\chi}
        }{\chi}
    \right].
    \label{eq:NA_dimensionless}
\end{align}

For $\Lambda\gg1$, the integral is dominated by $\chi=\mathcal O(1)$.
At leading logarithmic order,
\begin{equation}
    \log\left[
        \frac{
            2\Lambda^2-\chi
            +2\Lambda\sqrt{\Lambda^2-\chi}
        }{\chi}
    \right]
    =
    2\log(2\Lambda)-\log\chi
    +\mathcal O(\Lambda^{-2}),
    \label{eq:largeLambda_log}
\end{equation}
and, for $M\gg M_h$, the Higgs-mass term in the denominator may be
neglected. This gives
\begin{align}
    \frac{N}{A}
    \simeq{}&
    \frac{4}{\pi^2}
    \frac{
        \Gamma(1+\beta_\Gamma)
        \Gamma(3-\alpha_\Gamma)
    }{
        \Gamma(4-\alpha_\Gamma+\beta_\Gamma)
    }
    \frac{\Gamma_0v_\phi^2}{M^{2-2\alpha_\Gamma}}
    \left[
        \frac{V'(2v_\phi)}{M^2v_\phi}
    \right]^2
    \log(2\Lambda).
    \label{eq:NA_general_approx}
\end{align}
This expression is valid for $\alpha_\Gamma<3$, consistent with the
condition for threshold domination of the number integral.

\subsection{Fermion-pair production}
\label{subsec:fermion_production}

Consider a fermion coupled to $\phi$ through
\begin{equation}
    \mathcal L_{\rm int}=-y_\psi\phi\bar\psi\psi.
    \label{eq:Yukawa_interaction}
\end{equation}
For the process
\begin{equation}
    \phi^\ast\rightarrow\bar\psi\psi,
\end{equation}
the threshold lies at
\begin{equation}
    M=2m_\psi,
\end{equation}
and the spectral function is characterized by
\begin{equation}
    \alpha_\Gamma=1,
    \qquad
    \beta_\Gamma=\frac32.
\end{equation}

Equation~\eqref{eq:NA_heavy} then gives
\begin{align}
    \frac{N_\psi}{A}
    \simeq
    \frac{\bigl[V'(2v_\phi)\bigr]^2}{\pi^2}
    \int_0^{\gamma_w^2/l_w^2}\!\dd\chi\,
    \frac{\operatorname{Im}
    \widetilde{\Gamma}^{(2)}(\chi)}
    {\chi^2(\chi-M_h^2)^2}
    \log\left[
        \frac{
            2\gamma_w^2/l_w^2-\chi
            +2(\gamma_w/l_w)
             \sqrt{\gamma_w^2/l_w^2-\chi}
        }{\chi}
    \right].
    \label{eq:NoverA_general}
\end{align}

It is convenient to define
\begin{equation}
    \Lambda_\psi\equiv
    \frac{\gamma_w}{2m_\psi l_w}
    =\frac{E_{\max}}{m_\psi},
    \label{eq:Lambda_psi}
\end{equation}
where the last equality defines the maximum single-particle energy,
up to the convention used for $E_{\max}$. After rescaling the
integration variable by $4m_\psi^2$, we obtain
\begin{align}
    \frac{N_\psi}{A}
    \simeq{}&
    \frac{
        \bigl[V'(2v_\phi)\bigr]^2y_\psi^2
    }{
        8\pi^3(4m_\psi^2)^2
    }
    \int_1^{\Lambda_\psi^2}\!\dd\chi\,
    \frac{
        \left(1-\chi^{-1}\right)^{3/2}
    }{
        \chi^3
        \left[
            1-\dfrac{M_h^2}{4m_\psi^2\chi}
        \right]^2
    }
    \log\left[
        \frac{
            2\Lambda_\psi^2-\chi
            +2\Lambda_\psi
             \sqrt{\Lambda_\psi^2-\chi}
        }{\chi}
    \right].
    \label{eq:NA_integral}
\end{align}

In the limit
\begin{equation}
    \Lambda_\psi\gg1,
    \qquad
    m_\psi\gg M_h,
\end{equation}
the logarithm may be evaluated at leading order and the upper limit
extended to infinity:
\begin{align}
    \frac{N_\psi}{A}
    \simeq{}&
    \frac{
        \bigl[V'(2v_\phi)\bigr]^2y_\psi^2
    }{
        4\pi^3(4m_\psi^2)^2
    }
    \log(2\Lambda_\psi)
    \int_1^\infty\!\dd\chi\,
    \frac{\left(1-\chi^{-1}\right)^{3/2}}{\chi^3}.
    \label{eq:NoverA_leading_log}
\end{align}
Using
\begin{equation}
    \int_1^\infty\!\dd\chi\,
    \frac{\left(1-\chi^{-1}\right)^{3/2}}{\chi^3}
    =
    \frac{4}{35},
    \label{eq:fermion_integral}
\end{equation}
we find
\begin{equation}
    \frac{N_\psi}{A}
    \simeq
    \frac{
        \bigl[V'(2v_\phi)\bigr]^2
    }{
        35\pi^3(4m_\psi^2)^2
    }
    y_\psi^2
    \log(2\Lambda_\psi),
    \label{eq:NA_final}
\end{equation}
which is consistent with the results of Ref.~\cite{Ghoshal:2026pew} at the $y_\psi^2$ level up to an $O(1)$ factor. 

Although the maximum kinematically accessible energy can be much
larger than $m_\psi$, the integral is threshold dominated. The
produced fermions therefore typically have energies of order
\begin{equation}
    E_\psi=\mathcal O(m_\psi).
\end{equation}

\subsection{Cosmological number density and yield}
\label{subsec:fermion_yield}

To convert the number produced per unit collision area into a physical
number density, we use the geometric collision-area density
\begin{equation}
    \frac{A}{V}\simeq\frac{3}{R_*}.
    \label{eq:area_density}
\end{equation}

Following Ref.~\cite{Cataldi:2024pgt}, the characteristic bubble
separation is
\begin{align}
    R_*
    =
    \frac{(8\pi)^{1/3}}{\beta_{\rm PT}}
    =
    \frac{3\sqrt{5}}{\pi^{7/6}}\,
    v_\phi^{-2}
    \left[
        \frac{\pi^2\alpha_{\rm PT}}
        {30(1+\alpha_{\rm PT})c_V}
    \right]^{1/2}
    \frac{M_{\rm Pl}}{\beta_{\rm PT}/H}.
    \label{eq:Rstar}
\end{align}

Here $\alpha_{\rm PT}$ denotes the phase-transition strength and
$\beta_{\rm PT}$ its inverse duration. This notation distinguishes
$\alpha_{\rm PT}$ from the spectral exponent $\alpha_\Gamma$
introduced in Eq.~\eqref{eq:ImGamma_parametrization}.

Substituting Eq.~\eqref{eq:Rstar} into
Eq.~\eqref{eq:NA_final} gives
\begin{align}
    n_\psi
    ={}&
    \frac{
        \pi^{1/6}\bigl[V'(2v_\phi)\bigr]^2
    }{
        35\pi^2\sqrt{5}\,
        v_\phi^2(4m_\psi^2)^2
    }
    y_\psi^2
    \frac{\beta_{\rm PT}}{H}
    \left[
        \frac{
            30(1+\alpha_{\rm PT})c_V
        }{
            \pi^2\alpha_{\rm PT}
        }
    \right]^{1/2}
    \frac{v_\phi^4}{M_{\rm Pl}}
    \log\left(\frac{E_{\max}}{m_\psi}\right).
    \label{eq:npsi_cosmological}
\end{align}

For the potential
\begin{equation}
    V(\phi)
    =
    \frac{\lambda}{4}\phi^2(\phi-v_\phi)^2
    +c_Vv_\phi\phi^2(2\phi-3v_\phi),
    \label{eq:quartic_potential_production}
\end{equation}
one has
\begin{equation}
    V'(2v_\phi)
    =
    3(\lambda+4c_V)v_\phi^3.
    \label{eq:Vprime}
\end{equation}
The fermion number density therefore becomes
\begin{align}
    n_\psi
    ={}&
    \frac{
        9\pi^{1/6}(\lambda+4c_V)^2
    }{
        560\pi^2\sqrt{5}
    }
    y_\psi^2
    \left(\frac{v_\phi}{m_\psi}\right)^4
    \frac{\beta_{\rm PT}}{H}
    \left[
        \frac{
            30(1+\alpha_{\rm PT})c_V
        }{
            \pi^2\alpha_{\rm PT}
        }
    \right]^{1/2}
    \frac{v_\phi^4}{M_{\rm Pl}}
    \log\left(\frac{E_{\max}}{m_\psi}\right).
    \label{eq:npsi_lambda}
\end{align}

The scalar masses in the true and false vacua are
\begin{equation}
    M_t^2=
    \frac{\lambda+12c_V}{2}v_\phi^2,
    \qquad
    M_f^2=
    \frac{\lambda-12c_V}{2}v_\phi^2.
    \label{eq:false_true_masses}
\end{equation}
These relations imply
\begin{equation}
    (\lambda+4c_V)v_\phi^2
    =
    \frac{2}{3}\left(M_f^2+2M_t^2\right).
    \label{eq:mass_identity}
\end{equation}
Equation~\eqref{eq:npsi_lambda} may thus be written as
\begin{align}
    n_\psi
    ={}&
    \frac{9\pi^{1/6}}
         {140\pi^2\sqrt{5}}\,
    y_\psi^2
    \left(
        \frac{M_f^2+2M_t^2}{3m_\psi^2}
    \right)^2
    \frac{\beta_{\rm PT}}{H}
    \left[
        \frac{
            30(1+\alpha_{\rm PT})c_V
        }{
            \pi^2\alpha_{\rm PT}
        }
    \right]^{1/2}
    \frac{v_\phi^4}{M_{\rm Pl}}
    \log\left(\frac{E_{\max}}{m_\psi}\right).
    \label{eq:npsi_masses}
\end{align}

Dividing by the entropy density after reheating gives the fermion
yield,
\begin{align}
    Y_\psi
    \equiv\frac{n_\psi}{s}
    ={}&
    \frac{81\pi^{1/6}}
         {56\pi^4\sqrt{5}}\,
    y_\psi^2
    \left(
        \frac{M_f^2+2M_t^2}{3m_\psi^2}
    \right)^2
    \frac{\beta_{\rm PT}}{H}
    \left[
        \frac{
            \pi^2\alpha_{\rm PT}
        }{
            30(1+\alpha_{\rm PT})g_*c_V
        }
    \right]^{1/4}
    \frac{v_\phi}{M_{\rm Pl}}
    \log\left(\frac{E_{\max}}{m_\psi}\right).
    \label{eq:Ypsi}
\end{align}
In obtaining Eq.~\eqref{eq:Ypsi}, we assume prompt conversion of the
released vacuum energy into a thermal bath and take
$g_{*s}\simeq g_*$. If thermalization is delayed, or if only a
fraction of the released energy enters the thermal bath, the entropy
density and any intervening cosmological dilution must be treated
separately.

%% file: sec4.tex
\section{Particle Production in \texorpdfstring{$(3+1)$}{3+1} Dimensions}
\label{sec:particle}

\subsection{General Case}

In $(3+1)$ dimensions, the Fourier transform of the scalar field is defined as
\begin{align}
    \tilde{\phi}(\omega,{\bf k})
    =
    \int \d t\,\d^3{\bf x}\,
    e^{\rmi\omega t-\rmi{\bf k}\cdot{\bf x}}
    \phi(t,{\bf x}).
\end{align}
In the high-energy limit, the equation of motion gives
\begin{align}
    (\omega^2-{\bf k}^2)^2\tilde{\phi}(\omega,{\bf k})
    &=
    \int \d t\,\d^3{\bf x}\,
    e^{\rmi\omega t-\rmi{\bf k}\cdot{\bf x}}
    (\partial^2)^2\phi(t,{\bf x})
    \nnn
    &=
    -\int \d t\,\d^3{\bf x}\,
    e^{\rmi\omega t-\rmi{\bf k}\cdot{\bf x}}
    \partial^2V^\prime(\phi).
\end{align}

For highly relativistic bubbles, the bubble walls can be approximated by step functions. We decompose the field configuration as
\begin{align}
    \phi
    =
    \sum_i\phi_{{\rm d},i}\chi_{B_i}+\phi_{\rm c}
    \equiv
    \phi_{\rm d}+\phi_{\rm c},
\end{align}
where $\phi_{{\rm d},i}$ and $\phi_{\rm c}$ are smooth functions, and
$\chi_{B_i}$ is the characteristic function of the bubble domain $B_i$.
The derivative of the potential can then be decomposed as
\begin{align}
    V^\prime(\phi)
    =
    \sum_{a\in\{0,1\}^n}
    V^\prime\left(
        \phi_{\rm c}+\sum_i a_i\phi_{{\rm d},i}
    \right)
    \prod_i\chi_{B_i^{a_i}},
\end{align}
where $B_i^1=B_i$ and $B_i^0=B_i^c$, with $B_i^c$ denoting the complement
of $B_i$.

In the high-energy limit, the leading contribution arises from the
singular terms in $\partial^2V^\prime(\phi)$ generated by derivatives of
the characteristic functions. Since $\phi$ satisfies the equation of
motion away from the bubble nucleation points, the singular part of
$\partial^2\phi$ must vanish. For each $i$, this requires
\begin{align}
    \left[
        2(\partial_\mu\phi_{{\rm d},i})
        (\partial^\mu\chi_{B_i})
        +
        \phi_{{\rm d},i}\partial^2\chi_{B_i}
    \right]_{\rm sing}
    =
    0.
\end{align}
The leading collision-induced contribution is therefore
\begin{align}
    \bigl(\partial^2V^\prime(\phi)\bigr)_{\rm coll}
    \eqqcolon
    \sum_{j\neq k}
    (\partial_\mu\chi_{B_j})
    (\partial^\mu\chi_{B_k})
    \bar{V}_{j,k}^\prime(\phi),
\end{align}
as derived in Appendix~\ref{appendixC}.

Assuming that particle production is dominated by bubble-wall collisions,
the high-energy spectrum is given by
\begin{align}
    (\omega^2-{\bf k}^2)^2\tilde{\phi}(\omega,{\bf k})
    &=
    -\int \d t\,\d^3{\bf x}\,
    \sum_{j\neq k}
    (\partial_\mu\chi_{B_j})
    (\partial^\mu\chi_{B_k})
    \bar{V}_{j,k}^\prime(\phi)
    e^{\rmi\omega t-\rmi{\bf k}\cdot{\bf x}}
    \nnn
    &=
    -2\sum_{j<k}
    \int_{\partial B_j\cap\partial B_k}
    \d S\,
    n_{j,\mu}n_k^\mu
    \bar{V}_{j,k}^\prime(\phi)
    e^{\rmi\omega t-\rmi{\bf k}\cdot{\bf x}},
\end{align}
where $n_{j,\mu}$ and $n_{k,\mu}$ are the normal vectors to
$\partial B_j$ and $\partial B_k$, respectively, with their directions
fixed by the orientations of the corresponding surfaces. The source of
high-energy particle production is thus localized at the intersection of
the bubble walls, and its strength depends on the field configuration at
that intersection. In the following subsection, we analyze the collision
of two bubbles, whose walls intersect along a circle in three-dimensional
space.

\subsection{A Simple Example: Two-Bubble Collision}

Consider a system containing exactly two bubbles. We place their centers
on the $x$ axis and choose the time of first contact to be $t=0$. The
bubble domains are then
\begin{align}
    B_1
    &=
    \left\{
        (x-r_1)^2+y^2+z^2<(t+r_1)^2
    \right\},
    \\
    B_2
    &=
    \left\{
        (x+r_2)^2+y^2+z^2<(t+r_2)^2
    \right\},
\end{align}
as illustrated in Fig.~\ref{fig:collision}. Their characteristic
functions satisfy
\begin{align}
    \partial_\mu\chi_{B_1}
    &=
    (1,-{\bf u})\,
    \delta(t+r_1-R_1),
    \\
    \partial_\mu\chi_{B_2}
    &=
    (1,-{\bf v})\,
    \delta(t+r_2-R_2),
\end{align}
where
\begin{align}
    R_1
    &=
    \sqrt{(x-r_1)^2+y^2+z^2},
    \nnn
    R_2
    &=
    \sqrt{(x+r_2)^2+y^2+z^2},
    \nnn
    {\bf u}
    &=
    \frac{(x-r_1,y,z)}{R_1},
    \nnn
    {\bf v}
    &=
    \frac{(x+r_2,y,z)}{R_2}.
\end{align}

Since
\begin{align}
    V^\prime(\phi(x_{--}))
    &=
    V^\prime(0)
    =
    0,
    \\
    V^\prime(\phi(x_{+-}))
    &=
    V^\prime(\phi(x_{-+}))
    =
    V^\prime(v_\phi)
    =
    0,
\end{align}
we obtain
\begin{align}
    \bar{V}_{1,2}^\prime(\phi)
    =
    V^\prime(\phi(x_{++}))
    =
    V^\prime(2v_\phi).
\end{align}
The Fourier-space field profile is therefore
\begin{align}
    \tilde{\phi}(\omega,{\bf k})
    &=
    -\frac{2V^\prime(2v_\phi)}
    {(\omega^2-{\bf k}^2)^2}
    \int \d t\,\d^3{\bf x}\,
    (1-{\bf u}\cdot{\bf v})
    \delta(t+r_1-R_1)
    \delta(t+r_2-R_2)
    e^{\rmi\omega t-\rmi{\bf k}\cdot{\bf x}}
    \nnn
    &=
    -\frac{2V^\prime(2v_\phi)}
    {(\omega^2-{\bf k}^2)^2}
    \int
    \frac{2\pi R\,\d R}
    {\sqrt{1+R^2/(r_1r_2)}}
    e^{\rmi f(R)
    [\omega(r_1+r_2)+k_x(r_1-r_2)]}
    {\rm J}_0(k_rR),
\end{align}
where
\begin{align}
    f(R)
    &=
    \frac{1}{2}
    \left(
        \sqrt{1+\frac{R^2}{r_1r_2}}-1
    \right),
    \\
    k_r
    &=
    \sqrt{k_y^2+k_z^2}.
\end{align}
Here $k_x$ denotes the momentum component along the collision axis, while
$k_r$ is the magnitude of the transverse momentum. The variable $R$ is
the radius of the intersection circle and satisfies
\begin{align}
    R^2
    =
    (t+r_1)^2-(x-r_1)^2
    =
    (t+r_2)^2-(x+r_2)^2.
\end{align}

\input{fig_bubble_coll}

The efficiency factor is defined by
\begin{align}
    F(\chi)
    =
    \int
    \frac{\d\omega\,\d^3{\bf k}}{(2\pi)^4}\,
    \left|\tilde{\phi}(\omega,{\bf k})\right|^2
    \delta\left(\chi-\omega^2+{\bf k}^2\right).
\end{align}
Because the delta function is Lorentz invariant, we may perform a boost
along the $x$ direction such that
\begin{align}
    \omega(r_1+r_2)+k_x(r_1-r_2)
    \longrightarrow
    2\sqrt{r_1r_2}\,\omega.
\end{align}
For $k_rR\gg1$, the product of Bessel functions can be approximated by
\begin{align}
    RR^\prime k_r
    {\rm J}_0(k_rR){\rm J}_0(k_rR^\prime)
    \simeq
    \frac{\sqrt{RR^\prime}}{\pi}
    \cos\left[k_r(R-R^\prime)\right].
\end{align}
The efficiency factor then becomes
\begin{align}
    F(\chi)
    =
    \frac{4\bigl[V^\prime(2v_\phi)\bigr]^2}{\chi^4}
    \int
    \frac{
        h(R,R^\prime)\sqrt{RR^\prime}\,
        \d R\,\d R^\prime
    }{
        \sqrt{1+R^2/(r_1r_2)}
        \sqrt{1+(R^\prime)^2/(r_1r_2)}
    },
\end{align}
where
\begin{align}
    h(R,R^\prime)
    &=
    \int
    \frac{\d\omega\,\d k_x\,\d k_r}{(2\pi)^2}\,
    e^{
        \rmi[g(R)-g(R^\prime)]\omega
        -\rmi(R-R^\prime)k_r
    }
    \delta\left(
        \chi-\omega^2+k_x^2+k_r^2
    \right),
    \\
    g(R)
    &=
    \sqrt{r_1r_2+R^2}-\sqrt{r_1r_2}.
\end{align}

This integral is highly oscillatory, and its dominant contribution arises
from the region $R\simeq R^\prime$. Denoting the upper limit of the
$R$ integration by $R_0$, we obtain
\begin{align}
    F(\chi)
    &=
    \frac{4\bigl[V^\prime(2v_\phi)\bigr]^2}{\chi^4}
    \int
    \frac{
        R\,\d R\,\d\omega\,\d k_x\,\d k_r
    }{
        2\pi[1+R^2/(r_1r_2)]
    }
    \delta\bigl(g^\prime(R)\omega-k_r\bigr)
    \delta\left(
        \chi-\omega^2+k_x^2+k_r^2
    \right)
    \nnn
    &=
    \frac{4\bigl[V^\prime(2v_\phi)\bigr]^2}{\chi^4}
    \int
    \frac{
        2\pi R\,\d R
    }{
        [1+R^2/(r_1r_2)]
        \sqrt{1-[g^\prime(R)]^2}
    }
    \int
    \frac{\d\omega\,\d k}{(2\pi)^2}
    \delta\left(\chi-\omega^2+k^2\right)
    \nnn
    &=
    2\pi r_1r_2
    \left(
        \sqrt{1+\frac{R_0^2}{r_1r_2}}-1
    \right)
    \frac{4\bigl[V^\prime(2v_\phi)\bigr]^2}{\chi^4}
    \int
    \frac{\d\omega\,\d k}{(2\pi)^2}
    \delta\left(\chi-\omega^2+k^2\right).
\end{align}
The efficiency factor per unit collision area,
$F(\chi)/(\pi R_0^2)$, approaches the parallel-wall result $f(\chi)$ in
the limit $R_0\to0$.

For a system containing multiple bubbles, the preceding calculation
applies to the collision of any pair $B_i$ and $B_j$, provided that both
walls are still propagating through the vacuum. In this case, $R_0$ is
the radius of the intersection circle when a third bubble $B_k$ first
reaches the colliding walls, namely when
\begin{align}
    B_i\cap B_j\cap B_k\neq\emptyset.
\end{align}

Typically, $r_1$ and $r_2$ are both of order $R_\ast$, the mean bubble
radius at the time of collision, and $R_0$ is likewise of order
$R_\ast$. Assuming that each bubble collides with approximately six
neighboring bubbles, the total efficiency factor is estimated as
\begin{align}
    F(\chi)
    &\sim
    \frac{
        6\times2(\sqrt{2}-1)\pi R_\ast^2
    }{
        4\pi R_\ast^2
    }
    F_0(\chi)
    \nnn
    &=
    3(\sqrt{2}-1)F_0(\chi).
\end{align}
The average volume assigned to each bubble is approximately
\begin{align}
    V_0=(2R_\ast)^3,
\end{align}
whereas the volume used in the parallel-wall estimate is
\begin{align}
    V_\ast=\frac{4\pi}{3}R_\ast^3.
\end{align}
Consequently, the efficiency factor per unit volume is approximately
\begin{align}
    \frac{F(\chi)}{V_0}
    \sim
    0.65\,
    \frac{F_0(\chi)}{V_\ast}.
\end{align}
Thus, accounting for the finite bubble radius introduces an order-one
suppression of the total particle number per unit volume relative to the
parallel-wall approximation.

%% file: fig_bubble_coll.tex
\begin{figure}[t]     
    \centering
    \begin{tikzpicture}

        \def\Lr{6.0}
        \def\Lh{4.0}
        \def\dl{0.2}

        \coordinate     (B1)    at  (0,0);
        \coordinate     (B2)    at  (\Lr,0);
        \coordinate     (B3)    at  (0,\Lh);
        \coordinate     (B4)    at  (\Lr,\Lh);
        \coordinate     (X)     at  (\Lr/2+0.5,\Lh/2+1.0);
        \coordinate     (O)     at  (\Lr/2,\Lh/2);
        \coordinate     (Bk)    at  (\Lr,\Lh/2);
        \coordinate     (Bj)    at  (0,\Lh/2);
        \coordinate     (Bi)    at  (\Lr/2,0);
        \coordinate     (Bc)    at  (\Lr/2,\Lh);
        \coordinate     (Xpp)     at  (\Lr/2,\Lh/2-\dl);
        \coordinate     (Xpm)     at  (\Lr/2-\dl,\Lh/2);
        \coordinate     (Xmp)     at  (\Lr/2+\dl,\Lh/2);
        \coordinate     (Xmm)     at  (\Lr/2,\Lh/2+\dl);

        \fill           (O)     circle (1.2pt);
        \fill           (Xpp)   circle (1.2pt);
        \fill           (Xpm)   circle (1.2pt);
        \fill           (Xmp)   circle (1.2pt);
        \fill           (Xmm)   circle (1.2pt);

        \draw[thick]    (B1) -- (B4);
        \draw[thick]    (B2) -- (B3);
        \draw[blue,->]       (X)  -- (O);

        \node[]             at (Bj) {$B_j \cap B_k^c$};
        \node[]             at (Bk) {$B_j^c \cap B_k$};
        \node[]             at (Bi) {$B_j \cap B_k$};
        \node[]             at (Bc) {$B_j^c \cap B_k^c$};

        \node[above,blue]   at (X) {$x$};
        \node[below]        at (Xpp) {$x_{++}$};
        \node[left ]        at (Xpm) {$x_{+-}$};
        \node[right]        at (Xmp) {$x_{-+}$};
        \node[above]        at (Xmm) {$x_{--}$};

    \end{tikzpicture}
    \caption{The intersection of two bubble walls. $x_{ij}$ are four points in the four domains $B_j \cap B_k, B_j \cap B_k^c, B_j^c \cap B_k, B_j^c \cap B_k^c$ respectively, and they are close to $x \in \partial B_j \cap \partial B_k$. The value of $\bar{V}_{j,k}^\prime(\phi)$ depends on the field configuration at these four points.}
    \label{fig:intersection}
\end{figure}
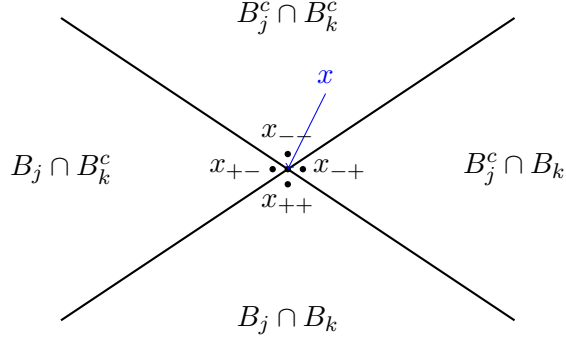

%% file: sec5_discussion.tex
\section{Discussion}
\label{sec:discussion}

In this paper, we investigated particle production during bubble collisions in a first-order phase transition. We found that, in the high-frequency regime (\textit{i.e.}, $\chi \equiv \omega^2-{\bf k}^2 \gg v_\phi^2$), the efficiency factor for both elastic and inelastic collisions behaves as
\begin{align}
    F(\chi) \propto
    \frac{\bigl[V^\prime(2v_\phi)\bigr]^2}{\chi^4}
    \theta\left(\frac{\gamma_w^2}{l_w^2}-\chi\right)
    \log\left[
        \frac{
            2\gamma_w^2/l_w^2-\chi
            +2(\gamma_w/l_w)\sqrt{\gamma_w^2/l_w^2-\chi}
        }{\chi}
    \right].
\end{align}
For inelastic collisions, this result is consistent with Ref.~\cite{Falkowski:2012fb}. For elastic collisions, however, the treatment in Ref.~\cite{Falkowski:2012fb} yields $F(\chi)\sim O(\chi^{-2})$, in contrast to our result. This discrepancy primarily stems from the assumption made in Ref.~\cite{Falkowski:2012fb} that the $\phi$ field instantaneously becomes spatially uniform and equal to the false-vacuum value $v_{\rm f}$ throughout the forward light cone, $t>0$ and $|x|<t$. In reality, the field relaxes from $2v_\phi$ to the vacuum over a finite timescale of order $O(v_\phi^{-1})$, as illustrated in Fig.~\ref{fig:profile}. Because the high-frequency behavior is particularly sensitive to abrupt changes in the field profile, the instantaneous-relaxation assumption leads to a different asymptotic behavior.

For particle production, the relevant scale $\chi\sim M_{\rm final}^2$ is much larger than the characteristic scale of the potential $V(\phi)$. Over the corresponding short timescale, the field $\phi$ can therefore be approximated as a free field, giving rise to a term proportional to $\delta(\omega^2-{\bf k}^2)$ in Fourier space. Furthermore, the singular part of the first derivative of $\phi$ is already encoded in the discontinuity of $\phi$, while its regular part remains continuous. The first derivative therefore does not contribute at leading order. Instead, the leading behavior is controlled by the discontinuity in the second derivative of $\phi$, which is proportional to $V^\prime(2v_\phi)$ and produces an asymptotic amplitude of order $O(\chi^{-2})$. Since the efficiency factor is proportional to the square of this contribution, its asymptotic behavior is
\begin{align}
    F(\chi)\sim O\left(\bigl[V^\prime(2v_\phi)\bigr]^2\chi^{-4}\right).
\end{align}

\begin{figure}[htbp]
    \centering
    \begin{subfigure}{0.48\textwidth}
        \centering
        \includegraphics[width=\textwidth]{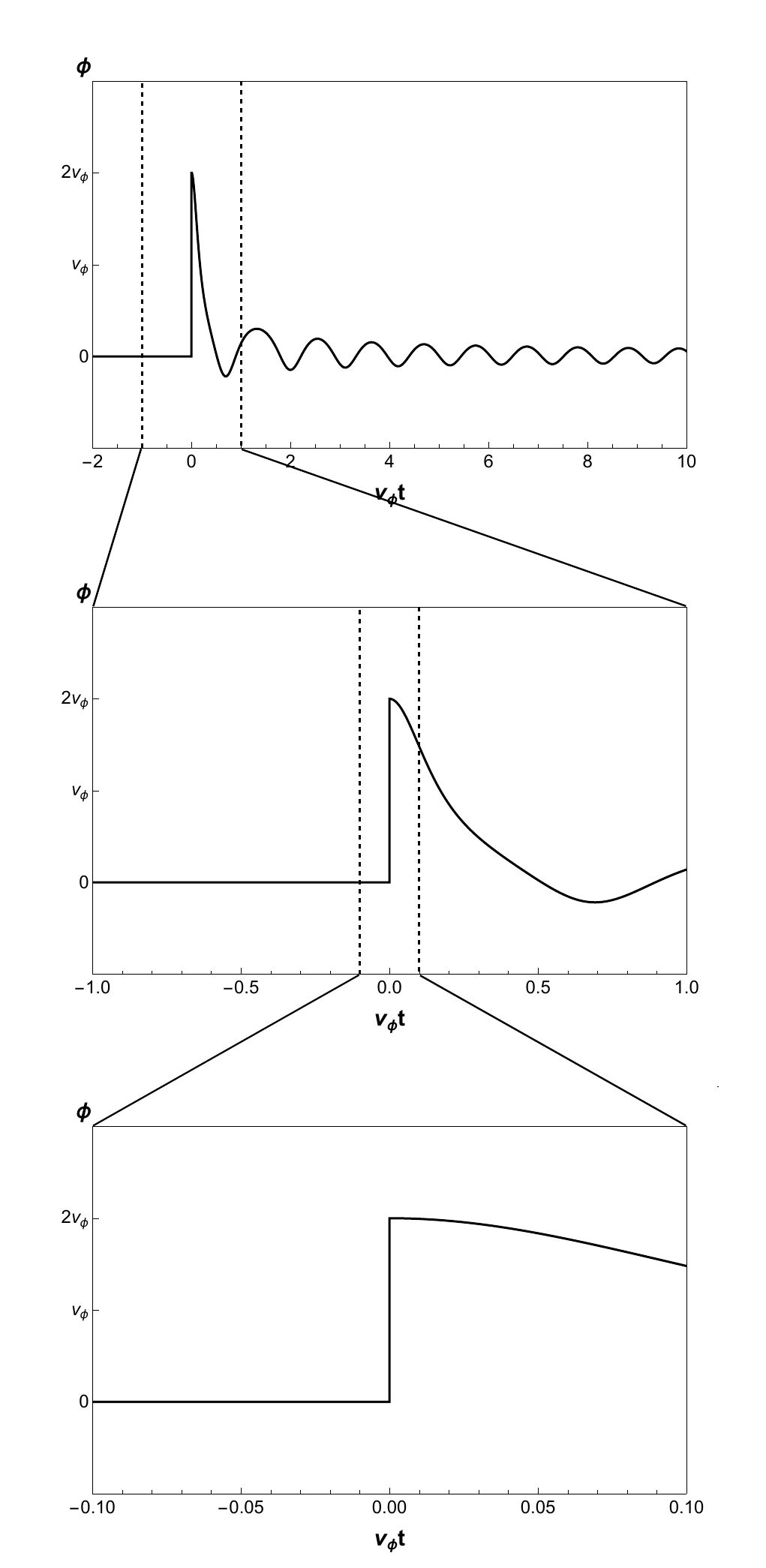}
        \caption{Elastic collision}
    \end{subfigure}
    \hfill
    \begin{subfigure}{0.48\textwidth}
        \centering
        \includegraphics[width=\textwidth]{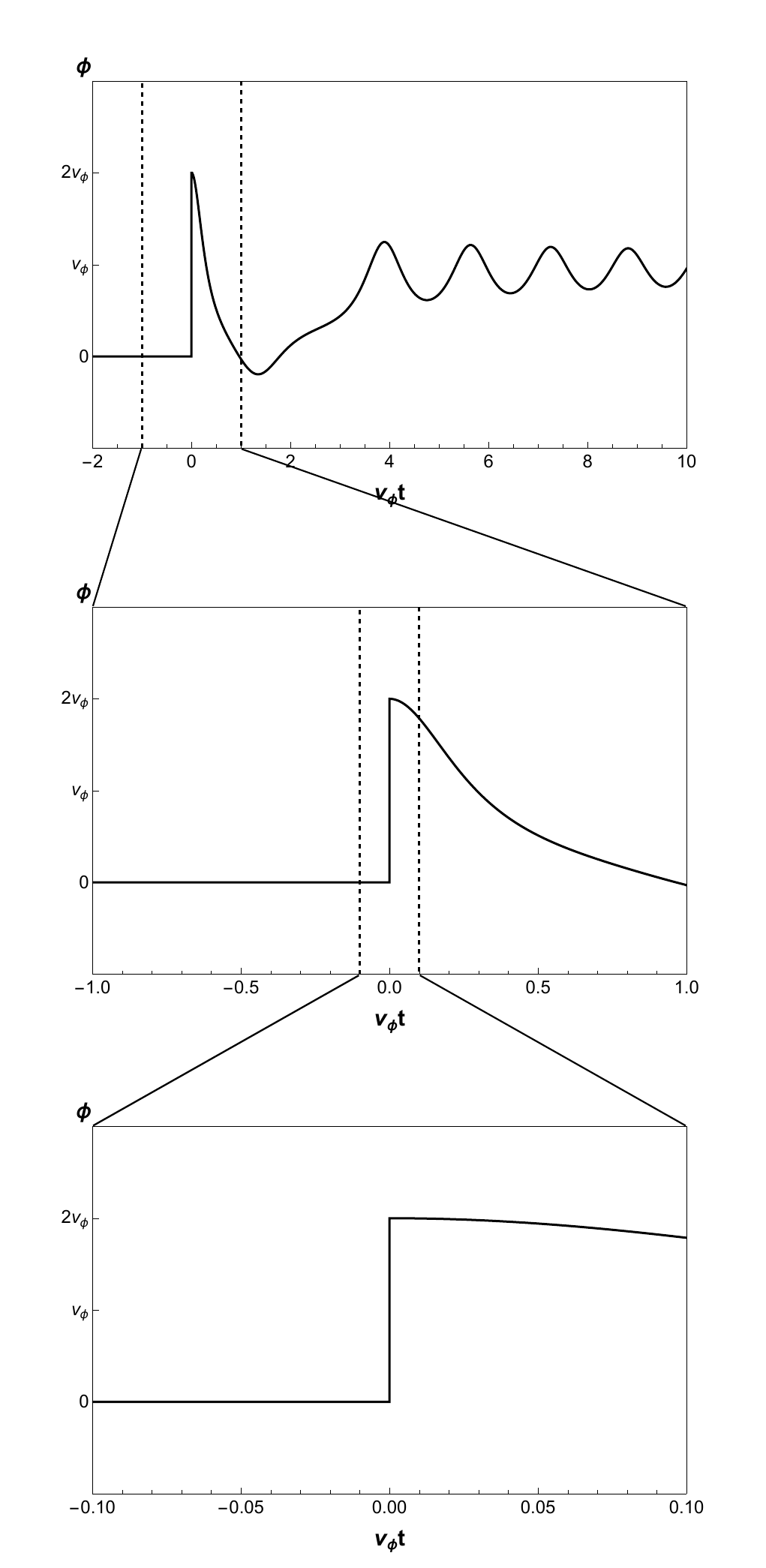}
        \caption{Inelastic collision}
    \end{subfigure}
    \caption{Field profiles $\phi(t)$ at $x=0$ for elastic (left) and inelastic (right) collisions. The solid curves show the field profiles, while the dashed curves denote auxiliary lines.}
    \label{fig:profile}
\end{figure}

The ultraviolet behavior observed in our simulations agrees with the analytical prediction. It differs, however, from that reported in Ref.~\cite{Mansour:2023fwj}. We attribute this discrepancy to our explicit treatment of spectral leakage arising from Fourier transforms performed over a finite interval, $[0,L]$. When extracting the high-frequency behavior, we introduce an $\rmi\epsilon$ prescription that provides exponential damping and suppresses spectral leakage. This procedure prevents the physical contribution, which scales as $O(\chi^{-2})$ at the amplitude level, from being contaminated in the ultraviolet by spurious $O(\chi^{-1})$ artifacts associated with the rectangular window function.

%% file: Acknowledgments.tex
\section*{Acknowledgments}

This work is supported in part by the National Key R\&D Program of China under Grants Nos. 2021YFC2203100 and 2017YFA0402204, the NSFC under Grant Nos. 12475107 and 12525506, and the Tsinghua University Dushi Program. 

%% file: appendixA.tex
\section{Large-\texorpdfstring{$s$}{s} Behavior of \texorpdfstring{$\h$}{h}}
\label{appendixA}

In this appendix, we show that $\h(s)$ decays as $s^{-1/2}$ in the limit $s\to\infty$.

We begin by treating the equation of motion for $\h(s)$ as a dynamical system with phase-space variables
\begin{align}
    (x,p)\coloneqq\bigl(\h(s),\h^\prime(s)\bigr).
\end{align}
Although the equation of motion contains a friction term proportional to $s^{-1}$, the fact that this coefficient vanishes as $s\to\infty$ does not by itself guarantee that $\h(s)$ converges. Nevertheless, the results of Refs.~\cite{Markus+2016+17+30,cabot2009long} imply that $\h(s)$ approaches zero. Rather than giving the full proof, we summarize the main argument.

Define the Lyapunov function
\begin{align}
    E(x,p)=\frac{1}{2}p^2+V(x).
\end{align}
Along a solution, it satisfies
\begin{align}
    \frac{\d E}{\d s}=-\frac{p^2}{s}\leq 0,
    \qquad s>0.
\end{align}
Thus the energy decreases monotonically, and
\begin{align}
    E(s)\leq E(s_0)
\end{align}
for every $s>s_0>0$. Under our assumptions on the potential, this implies that $\h(s)$ remains bounded for sufficiently large $s$.

The friction coefficient $s^{-1}$ decreases monotonically but is not integrable at infinity:
\begin{align}
    \int_{s_0}^{\infty}\frac{\d s}{s}=\infty.
\end{align}
If $p(s)$ did not approach zero, the system would lose an amount of energy proportional to $s^{-1}$ during each oscillation. Because the curvature of the potential at a local minimum is nonzero, the oscillation period remains finite. The cumulative energy loss would therefore diverge, contradicting the fact that the energy is bounded from below. Hence,
\begin{align}
    \h^\prime(s) = p(s)\longrightarrow 0
    \qquad\text{as}\qquad s\to\infty.
\end{align}
The limit set must consequently be contained in
\begin{align}
    \{(x,p)\mid p=0\}.
\end{align}
Because the limit set is invariant under the flow, it must in fact be contained in
\begin{align}
    \{(x,p)\mid p=0,\;V^\prime(x)=0\},
\end{align}
which is a finite set. It follows that $\h(s)$ approaches a constant as $s\to\infty$. From the definition and boundary conditions of $\h(s)$, this constant must be zero:
\begin{align}
    \h(s)\longrightarrow 0
    \qquad\text{as}\qquad s\to\infty.
\end{align}

We next determine the decay rate. Since the limits of $\h(s)$ and $\h^\prime(s)$ exist, we consider
\begin{align}
    s\bigl[E(s)-E(\infty)\bigr].
\end{align}
Using $\d E/\d s=-p^2/s$, we obtain
\begin{align}
    \frac{\d}{\d s}
    \left\{
        s\bigl[E(s)-E(\infty)\bigr]
    \right\}
    =
    -\frac{1}{2}p^2+V(x)-V\bigl(x(\infty)\bigr).
\end{align}

If the limiting configuration $x(\infty)$ is a local maximum of the potential, then
\begin{align}
    V(x)-V\bigl(x(\infty)\bigr)<0
\end{align}
for all sufficiently large $s$. Consequently,
\begin{align}
    \frac{\d}{\d s}
    \left\{
        s\bigl[E(s)-E(\infty)\bigr]
    \right\}<0.
\end{align}
This case is incompatible with convergence to the maximum unless the solution is fine-tuned to remain there. Therefore, for a nontrivial solution, the relevant asymptotic state is a local minimum.

Suppose now that $x(\infty)$ is a local minimum. After shifting the field and the potential, we may set
\begin{align}
    x(\infty)=0,
    \qquad
    E(\infty)=V(0)=0.
\end{align}
Using the equation of motion, the preceding identity can be rewritten as
\begin{align}
    \frac{\d}{\d s}\bigl[sE(s)\bigr]
    &=
    -\frac{1}{2}\frac{\d(xp)}{\d s}
    -\frac{1}{4}\frac{\d(x^2/s)}{\d s}
    -\frac{x^2}{4s^2}
    +V(x)-\frac{1}{2}xV^\prime(x).
\end{align}
Because $(x,p)$ is bounded, the integrated contributions of the first three terms on the right-hand side are bounded. Moreover, since the curvature of the potential at the minimum is positive, there exists a constant $C>0$ such that, sufficiently close to the minimum,
\begin{align}
    \left|
        V(x)-\frac{1}{2}xV^\prime(x)
    \right|
    \leq C E(s)^{3/2}.
\end{align}
As $(x,p)\to(0,0)$, this bound holds for all sufficiently large $s$. Hence there exist constants $s_0$ and $C^\prime>0$ such that
\begin{align}
    E(s)
    \leq
    \frac{(t+C^\prime)E(t)}{s}
    +\frac{s-t}{s}\,C E(t)^{3/2}
\end{align}
for all $s>t>s_0$. This inequality implies
\begin{align}
    E(s)=O(s^{-1})
    \qquad\text{as}\qquad s\to\infty.
\end{align}
Finally, since the energy is quadratic to leading order near the local minimum,
\begin{align}
    E(s)
    =
    \frac{1}{2}\bigl[\h^\prime(s)\bigr]^2
    +\frac{1}{2}V^{\prime\prime}(0)\h(s)^2
    +O\bigl(\h(s)^3\bigr),
\end{align}
we conclude that
\begin{align}
    \boxed{\h^\prime(s),\h(s)=O(s^{-1/2})}
    \qquad\text{as}\qquad s\to\infty.
\end{align}

%% file: appendixB.tex
\section{On the order of \texorpdfstring{$I_n(\omega,k)$}{}}\label{appendixB}

We consider the trapping equation
\begin{align}
    \h^{\prime\prime}(s) + \frac{1}{s} \h^{\prime}(s) + M_h^2 \h(s) = - \Delta(\h(s)) \ , 
\end{align}
where the function $\Delta(\h)$ is smooth in $\h$ and satisfies $\Delta(0) = \Delta^\prime(0) = 0$. In the previous section, we showed that $\h(s), \h^\prime(s) = O(s^{-1/2})$, and we will use this fact throughout this section without further proof. 

The operator $\LB$ is defined by
\begin{align}
    (\LB f)(s) = f^{\prime\prime}(s) + \frac{1}{s} f^{\prime}(s) \ . 
\end{align}

We first establish several properties of the solution $\h$. The trapping equation implies the following smoothness statement.
\begin{lemma} \label{lemma:smoothness}
    $\LB^n G\big(\h(s),(\h^\prime(s))^2\big)$ is smooth with respect to $s$ for any smooth $G$ and $n \in \mathbb{N}$. 
\end{lemma}

\begin{proof}
    Notice that the function $\h$ is the solution of the differential equation 
    \begin{align}
        \h^{\prime\prime}(s) + \frac{1}{s} \h^{\prime}(s) + M_h^2 \h(s) = - \Delta(\h(s)) \ . 
    \end{align}
    Let $u(x,y) = \h\Big(\sqrt{x^2 + y^2}\Big)$. The equation above can be rewritten as 
    \begin{align}
        (\partial_x^2 + \partial_y^2)u(x,y) + M_h^2 u(x,y) = - \Delta(u(x,y)) \ ,
    \end{align}
    Since $u$ is a bounded solution to a semilinear elliptic equation, elliptic regularity implies that $u(x,y)$ is smooth in $(x,y)$. Hence, for any $n \in \mathbb{N}$,
    \begin{align}
        \big(\LB^n G\big(\h(s),(\h^\prime(s))^2\big)\big)\Big|_{s = \sqrt{x^2 + y^2}} = (\partial_x^2 + \partial_y^2)^n G(u(x,y),|\nabla u(x,y)|^2)
    \end{align}
    is smooth with respect to $(x,y)$. Setting $x \geq 0$ and $y=0$, we conclude that the expression $\big(\LB^n G\big(\h(s),(\h^\prime(s))^2\big)\big)\Big|_{s = x}$ is smooth in $x$, which is exactly the desired result after relabeling the variable. 
\end{proof} 

The lemma above establishes regularity at finite $s$. We now turn to the behavior as $s \to \infty$. Let $Q \coloneqq C^\infty(\mathbb{R}^2)[x]$ be the ring of polynomials in $x$ with coefficients in the space of smooth functions on $\mathbb{R}^2$, and define $A_n \coloneqq \left\{ \sum_{p=0}^n g_p(\h(s),\h^\prime(s);s^{-1}) \h(s)^p (\h^\prime(s))^{n-p} : g_p \in Q \right\}$. We then have the following lemma.
\begin{lemma} \label{lemma:bound}
    $\LB A_n \subset A_n$ for any $n \in \mathbb{N}$. 
\end{lemma}

\begin{proof}
    Clearly, $s^{-1} A_n \subset A_n$, so it remains to prove $\partial_s A_n \subset A_n$. 
    Indeed,
    \begin{align}
        &\quad \partial_s \left( \sum_{p=0}^n g_p(\h(s),\h^\prime(s);s^{-1}) \h(s)^p (\h^\prime(s))^{n-p} \right) \nnn
        &=\sum_{p=0}^n g_p^{(1,0,0)}(\h(s),\h^\prime(s);s^{-1}) \h^\prime(s) \h(s)^p (\h^\prime(s))^{n-p} \nnn 
        &\quad + \sum_{p=0}^n g_p^{(0,1,0)}(\h(s),\h^\prime(s);s^{-1}) \h^{\prime\prime}(s) \h(s)^p (\h^\prime(s))^{n-p} \nnn
        &\quad - \sum_{p=0}^n g_p^{(0,0,1)}(\h(s),\h^\prime(s);s^{-1}) s^{-2} \h(s)^p (\h^\prime(s))^{n-p} \nnn 
        &\quad + \sum_{p=0}^n p g_p(\h(s),\h^\prime(s);s^{-1}) \h(s)^{p-1} (\h^\prime(s))^{n-p+1} \nnn
        &\quad + \sum_{p=0}^n g_p(\h(s),\h^\prime(s);s^{-1}) (n-p) \h^{\prime\prime}(s) \h(s)^p (\h^\prime(s))^{n-p-1} \ . 
    \end{align}

    The first, third, and fourth terms in the expression above belong to $A_n$. It therefore suffices to analyze the second and last terms. 

    Using the fact that 
    \begin{align}
        \h^{\prime\prime}(s) = - s^{-1} \h^\prime(s) - M_h^2 \h(s) - \Delta(\h(s)) \ , 
    \end{align}
    the second term is immediately seen to lie in $A_n$. 

    The last term can be written as 
    \begin{align}
        &\quad \sum_{p=0}^n g_p(\h(s),\h^\prime(s);s^{-1}) (n-p) \h^{\prime\prime}(s) \h(s)^p (\h^\prime(s))^{n-p-1} \nnn
        &= - \sum_{p=0}^n g_p(\h(s),\h^\prime(s);s^{-1}) s^{-1} (n-p) \h(s)^p (\h^\prime(s))^{n-p} \nnn
        &\quad - \sum_{p=0}^n g_p(\h(s),\h^\prime(s);s^{-1}) \bigg(M_h^2 + \frac{\Delta(\h)}{\h}\bigg) (n-p) \h(s)^{p+1} (\h^\prime(s))^{n-p-1} \ , 
    \end{align}
    so this term also belongs to $A_n$. 

    Therefore $\partial_s A_n \subset A_n$, which implies $\LB A_n = (\partial_s^2 + s^{-1} \partial_s) A_n \subset A_n$. 
\end{proof} 

For simplicity, we denote 
\begin{align}
    H_{i j k} \coloneqq \h(s)^i (\h^{\prime}(s))^j s^{-k} \ , 
\end{align}
and let $B_n = \left\{f + \sum_{i+j+2 k \geq n} C_{i j k} H_{i j k} : f \in A_n, |\{(i,j,k) : C_{ijk} \neq 0\}| < \infty \right\}$. 
We then obtain the following theorem.
\begin{theorem} \label{theorem:bound}
    $\LB B_n \subset B_n$ for any $n \in \mathbb{N}$. 
\end{theorem}

\begin{proof}
    Clearly, $s^{-1} B_n \subset B_n$, so it remains to show $\partial_s B_n \subset B_n$. Since $A_n \subset B_n$, Lemma~\ref{lemma:bound} gives $\partial_s A_n \subset A_n \subset B_n$. Hence the derivative of $f$ lies in $B_n$. 

    Next we consider the $H_{ijk}$ terms: 
    \begin{align}
        &\quad \partial_s \left( \sum_{i+j+2 k \geq n} C_{i j k} H_{i j k} \right) \nnn
        &=\sum_{i+j+2 k \geq n} i C_{i j k} H_{i-1,j+1,k}
        + \sum_{i+j+2 k \geq n} j C_{i j k} \h^{\prime\prime}(s) H_{i,j-1,k} 
        - \sum_{i+j+2 k \geq n} k C_{i j k} H_{i,j,k+1} 
    \end{align}

    The first and last terms are in $B_n$, so it suffices to treat the second term. 

    Using the fact that 
    \begin{align}
        \h^{\prime\prime}(s) = - s^{-1} \h^\prime(s) - M_h^2 \h(s) - \Delta(\h(s)) \ , 
    \end{align}
    the second term can be written as 
    \begin{align}
        &\quad \sum_{i+j+2 k \geq n} j C_{i j k} \h^{\prime\prime}(s) H_{i,j-1,k} \nnn
        &= - \sum_{i+j+2 k \geq n} j C_{i j k} H_{i,j,k+1} - \sum_{i+j+2 k \geq n} j C_{i j k} M_h^2 H_{i+1,j-1,k} \nnn
        &\quad - \sum_{i+j+2 k \geq n} \sum_{q = 2}^{n} \frac{j C_{i j k} }{q!} \Delta^{q}(0) H_{i+q,j-1,k} \nnn
        &\quad - \sum_{i+j+2 k \geq n} j C_{i j k} \frac{\Delta(\h) - \sum_{q = 2}^{n} \Delta^{q}(0) \h^q / q!}{ \h^{n+1}} H_{i+n+1,j-1,k} \ .  
    \end{align}
    The first three terms belong to $B_n$, while the last one lies in $A_n \subset B_n$. Hence the whole expression belongs to $B_n$. 

    Therefore $\partial_s B_n \subset B_n$, and consequently $\LB B_n = (\partial_s^2 + s^{-1} \partial_s) B_n \subset B_n$. 
\end{proof} 

This theorem leads to the following corollary.
\begin{corollary} \label{corollary:bound}
    Let $f \in B_4$ be a function, then $\LB^n f = O(s^{-2})$ for any $n \in \mathbb{N}$.  
\end{corollary}

\begin{proof}
    Since $\h,\h^\prime = O(s^{-1/2})$, we have $H_{ijk} = O(s^{-(i+j+2k)/2})$. Thus for $i + j + 2 k \geq 4$, we have $H_{ijk} = O(s^{-2})$. By the smoothness of $g_p$, the function $g_p(\h(s),\h^\prime(s);s^{-1})$ is bounded for large $s$, so $g_p(\h(s),\h^\prime(s);s^{-1}) H_{p,4-p,0} = O(s^{-2})$. Therefore, for any $f \in B_4$, we have $f = O(s^{-2})$. By Theorem~\ref{theorem:bound}, we have $\LB^n B_4 \subset B_4$, hence $\LB^n f = O(s^{-2})$. 
\end{proof}

Corollary~\ref{corollary:bound} shows that every $f \in B_4$ is $O(s^{-2})$ as $s \to \infty$. Consequently, $B_4 \cap L^1_{\rm loc}(0,\infty) \subset L^1(0,\infty)$.

Now consider the transformation 
\begin{align}
    F[f](\alpha) = \int_0^{\infty} 2 s \d s \ {\rm K}_0(\alpha s) f(s) \ , 
\end{align}
where we restrict $\alpha$ to have non-negative real part. Unless otherwise specified, $F$ is defined on the class of all functions $f$ such that $(s_0+s)^{-N} s^{1/2} f(s)\in L^1(0,\infty)$ for some $N > 0$ and $s_0 > 0$; this is a sufficient working domain for this note, though not the maximal one. We denote this class by $\F$, so $s^{1/2} \F \coloneqq \{ s^{1/2} f : f \in \F \} \subset L^1_{\rm loc}(0,\infty)$. Then the integral $F[f](\alpha)$ is well defined and analytic in $\alpha$ on $D = \{ \alpha \in \mathbb{C} : {\rm Re}(\alpha) > 0 \}$. We now prove the following property of the transform $F$.
\begin{theorem}\label{theorem:continuity}
    For $f \in \F$, we have $F[f](\rmi \omega) = \lim_{\epsilon \to 0^+} F[f](\epsilon + \rmi \omega)$ as generalized functions. 
\end{theorem}

\begin{proof}
    Let $\varphi(\omega)$ be a test function, and define $T_\epsilon(\omega) = F[f](\epsilon + \rmi \omega)$ as the corresponding generalized function. Then
    \begin{align}
        T_\epsilon(\varphi) = \int_0^{\infty} 2 s \d s \ f(s) \int_{-\infty}^{\infty} \d \omega \ {\rm K}_0((\epsilon + \rmi \omega) s) \varphi(\omega) \ . 
    \end{align}
    Exchanging the order of integration is harmless here, since the right-hand side is precisely the definition of the generalized function $T_\epsilon$. 

    Let $ {\rm Ki}_n(z)$ be defined by 
    \begin{align}
         {\rm Ki}_n(z) = \int_0^{\infty} \d t \ \frac{e^{- z \cosh t}}{\cosh^n t} \ ,
    \end{align} 
    which is the so-called Bickley--Naylor function. If $n \geq 1$, then $ {\rm Ki}_n$ is bounded and satisfies $\partial^n  {\rm Ki}_n(z) = (-1)^n {\rm K}_0(z)$. 

    For $\epsilon \geq 0$, one has
    \begin{align}
        &~\int_{-\infty}^{\infty} \d \omega \  {\rm Ki}_n((\epsilon + \rmi \omega) s) (\partial^n \varphi)(\omega) 
        = (-\rmi s)^{n} \int_{-\infty}^{\infty} \d \omega \ (\partial^n  {\rm Ki}_n)((\epsilon + \rmi \omega) s) \varphi(\omega) \nnn
        =& (\rmi s)^{n} \int_{-\infty}^{\infty} \d \omega \ {\rm K}_0((\epsilon + \rmi \omega) s) \varphi(\omega)  \ .
    \end{align}
    Therefore,
    \begin{align}
        \bigg| \int_{-\infty}^{\infty} \d \omega \ {\rm K}_0((\epsilon + \rmi \omega) s) \varphi(\omega) \bigg| \leq s^{-n} \|  {\rm Ki}_n \|_\infty \| \partial^n \varphi \|_1 \ .
    \end{align}

    Hence
    \begin{align}
        &~ |(T_\epsilon - T_0)(\varphi)| = \bigg| \int_0^{\infty} 2 s \d s \ f(s) \int_{-\infty}^{\infty} \d \omega \ ({\rm K}_0((\epsilon + \rmi \omega) s) - {\rm K}_0(\rmi \omega s)) \varphi(\omega) \bigg| \nnn
        \leq &~ 2 \int_0^{a} \d s \ |f(s)| \bigg| s \int_{-\infty}^{\infty} \d \omega \ ({\rm K}_0((\epsilon + \rmi \omega) s) - {\rm K}_0(\rmi \omega s)) \varphi(\omega) \bigg| \nnn
        &+ \int_a^{\infty} 2 s \d s \ |f(s)| \bigg| \int_{-\infty}^{\infty} \d \omega \ {\rm K}_0((\epsilon + \rmi \omega) s) \varphi(\omega) \bigg| 
        + \int_a^{\infty} 2 s \d s \ |f(s)| \bigg| \int_{-\infty}^{\infty} \d \omega \ {\rm K}_0(\rmi \omega s) \varphi(\omega) \bigg| \nnn
        \leq &~ 2 \int_0^{a} \d s \ |f(s)| \int_{-\infty}^{\infty} \d \omega \ |  {\rm Ki}_1((\epsilon + \rmi \omega) s) -  {\rm Ki}_1(\rmi \omega s) | \cdot |\partial \varphi(\omega) | \nnn
        &+ 4 \|  {\rm Ki}_n \|_\infty \| \partial^n \varphi \|_1 \int_a^{\infty} \d s \ |f(s)| s^{-n+1} \nnn
        \leq &~ 2 \|\partial \varphi \|_1 \int_0^{a} \d s \ \int_0^{\infty} \d t \ |f(s)| \frac{1 - e^{-\epsilon s \cosh{t}}}{\cosh{t}} + 4 \|  {\rm Ki}_n \|_\infty \| \partial^n \varphi \|_1 \int_a^{\infty} \d s \ |f(s)| s^{-n+1} \nnn
    \end{align}
    Choose $n$ large enough so that $s^{-n+1} f(s) \in L^1(a,\infty)$. For any $\delta > 0$, we can then take $a$ large enough that the second term is below $\delta/2$. Since $\lim_{\epsilon \to 0^+} |f(s)| (1 - e^{-\epsilon s \cosh{t}})/\cosh{t} = 0$ pointwise in $(s,t)$, and for $\epsilon < 1$ we have 
    \begin{align}
        0 \leq s^{1/2} |f(s)| \frac{1 - e^{-\epsilon s \cosh{t}}}{s^{1/2} \cosh{t}} \leq \frac{(s_0 + a)^N s^{1/2} |f(s)|}{(s_0 + s)^N \sqrt{\cosh{t}}} \in L^1((0,a) \times (0,\infty)) \ . 
    \end{align}
    By dominated convergence, we may choose $\epsilon$ small enough that the first term is below $\delta/2$. Hence $|(T_\epsilon - T_0)(\varphi)| < \delta$ for sufficiently small $\epsilon$, which implies $T_\epsilon \to T_0$ in the sense of generalized functions. 
\end{proof}

Now our goal is to compute 
\begin{align}
    &~ I_n = F[\LB^n \Delta] = \int_0^{\infty} 2 s \d s \  {\rm K}_0(\alpha s) \LB^n \Delta(\h(s)) \ . 
\end{align}

We first prove the following lemma.
\begin{lemma} \label{lemma:L1function}
    Let $f \in \F$ be a function on $[0,\infty)$ with $s^{1/2} f \in L^1(0,\infty)$. Then the function $F[f](\alpha)$ is continuous in $\alpha$ on $\bar{D} \setminus \{0\}$ and $F[f](\alpha) = O(|\alpha|^{-1/2})$ as $\alpha \to \infty$. 
\end{lemma}

\begin{proof}
    This follows directly from the definition of $F$:
    \begin{align}
        |F[f](\alpha)| \leq \int_0^{\infty} 2 s \d s \ |{\rm K}_0(\alpha s) f(s)| \leq 2 \| s^{1/2} {\rm K}_0(\alpha s) \|_\infty \| s^{1/2} f \|_1 \leq (2 \pi)^{1/2} |\alpha|^{-1/2} \| s^{1/2} f \|_1 \ . 
    \end{align}
    Thus $F[f](\alpha) = O(|\alpha|^{-1/2})$ as $\alpha \to \infty$. 

    The continuity of $F[f](\alpha)$ in $\alpha$ on $\bar{D} \setminus \{0\}$ follows from 
    \begin{align}
        &~|F[f](\alpha) - F[f](\beta)| \nnn
        \leq&~ 2 \| s^{1/2} ({\rm K}_0(\alpha s) - {\rm K}_0(\beta s)) \chi_{[0,L)} \|_\infty \| s^{1/2} \chi_{[0,L)} f \|_1 \nnn 
        &+ 2 \| s^{1/2} ({\rm K}_0(\alpha s) - {\rm K}_0(\beta s)) \chi_{[L,\infty)} \|_\infty \| s^{1/2} \chi_{[L,\infty)} f \|_1 \nnn
        \leq&~ 2 \| s^{1/2} ({\rm K}_0(\alpha s) - {\rm K}_0(\beta s)) \chi_{[0,L)} \|_\infty \| s^{1/2} f \|_1 \nnn
        &+ (2 \pi)^{1/2} (|\alpha|^{-1/2} + |\beta|^{-1/2}) \| s^{1/2} \chi_{[L,\infty)} f \|_1 \ . 
    \end{align}
    For any $\epsilon > 0$, choose $L$ large enough so that the second term is less than $\epsilon/2$. Then choose $\delta$ small enough so that the first term is less than $\epsilon/2$ whenever $|\alpha - \beta| < \delta$. Hence $|F[f](\alpha) - F[f](\beta)| < \epsilon$ for $|\alpha - \beta| < \delta$, which implies that $F[f](\alpha)$ is continuous in $\alpha$ on $\bar{D} \setminus \{0\}$. 
\end{proof}

For continuous functions $f$ satisfying suitable assumptions, the following relation between $F[\LB f]$ and $F[f]$ holds.
\begin{lemma} \label{lemma:IBP}
    Let $f \in \F$ be a continuous function on $[0,\infty)$ with $\LB f \in \F$. Then we have 
    \begin{align}
        F[\LB f](\alpha) = \alpha^2 F[f](\alpha) - 2 f(0) \ . 
    \end{align}
    for any $\alpha \in D$. 
\end{lemma}

\begin{proof}
    Decompose $f$ as $f_1 + f_2$, where $f_1(0) = 0$ and $f_2(s) = f(0) e^{-s^2}$. Then $\LB f_2 \in \F$ and $F[\LB f_2](\alpha) = \alpha^2 F[f_2](\alpha) - 2 f(0)$, so it is enough to prove the lemma for functions satisfying $f(0)=0$. 

    Let $f$ be a continuous function on $[0,\infty)$ with $f(0) = 0$ and $\LB f \in \F$. Let $g(s) = \LB f$. Then $s f^\prime(s) = C + \int_0^s \d t \, t g(t)$. Since $g \in \F$, we have $s^{1/2} g(s) \in L^1_{\rm loc}(0,\infty)$, and thus 
    \begin{align}
        |(f(s) - C \log{(s/s_0)})^\prime| \leq s^{-1} \int_0^s t |g(t)| \d t \leq s^{-1/2} \int_0^s t^{1/2} |g(t)| \d t \in L^1_{\rm loc}(0,\infty) \ . 
    \end{align}
    Thus $f(s) - C \log{(s/s_0)}$ is continuous on $[0,\infty)$. Since $f$ itself is continuous on $[0,\infty)$, we must have $C = 0$. Using the estimate
    \begin{align}
        s^{1/2} |f^\prime(s)| \leq \int_0^s \d t \ t^{1/2} |g(t)| \ , 
    \end{align}
    we find that $s^{1/2} f^\prime(s)$ is locally absolutely continuous on $[0,\infty)$ with $\lim_{s \to 0} s^{1/2} f^\prime(s) = 0$. Therefore $\lim_{s \to 0} {\rm K}_0(\alpha s) s f^\prime(s) = 0$. 

    On the other hand, 
    \begin{align}
        &~(s + s_0)^{-N} s^{1/2} |f^\prime(s)| \leq \int_0^s \d t \ t^{1/2} (s + s_0)^{-N} |g(t)| \nnn
        \leq&~ \int_0^\infty \d t \ t^{1/2} (t + s_0)^{-N} |g(t)| = \|(s + s_0)^{-N} s^{1/2} g(s)\|_1 < \infty 
    \end{align}
    shows that $f \in \F$ and $f^\prime(s) = O((s + s_0)^N)$. Then $\lim_{s \to \infty} {\rm K}_0(\alpha s) s f^\prime(s) = 0$ for ${\rm Re}(\alpha) > 0$. 

    By the same argument, $f$ is locally absolutely continuous on $[0,\infty)$ with $\lim_{s \to 0} f(s) = 0$. Therefore both boundary limits of $\alpha s {\rm K}_0(\alpha s) f(s)$, as $s\to0$ and $s\to\infty$, vanish.

    Integration by parts is therefore valid, with vanishing boundary terms, and yields
    \begin{align}
        &~F[\LB f](\alpha) = \lim_{a \to 0^+} \int_a^\infty 2 s \d s \ {\rm K}_0(\alpha s) \LB f(s) \nnn
        =&~ \lim_{a \to 0^+} \int_a^\infty 2 \d s \ {\rm K}_0(\alpha s) (s f^{\prime}(s))^\prime 
        = \lim_{a \to 0^+} \int_a^\infty 2 \alpha s \d s \ {\rm K}_1(\alpha s) f^{\prime}(s) \nnn
        =&~ \lim_{a \to 0^+} \int_a^\infty 2 \alpha^2 s \d s \ {\rm K}_0(\alpha s) f(s) 
        = \alpha^2 F[f](\alpha) \ . 
    \end{align}
    This completes the proof. 
\end{proof}

The lemma above directly leads to the following theorem.
\begin{theorem} \label{theorem:MultipleFunctions}
    Let $\{ f_i : f_i \in \F \}_{i=1}^N$ be a finite collection of continuous functions on $[0,\infty)$. Suppose that there is an $N \times N$ constant matrix $A$ such that 
    \begin{align}
        s^{1/2} h_i \coloneqq s^{1/2} \left(\LB f_i + \sum_{j=1}^N A_{ij} f_j \right) \in L^1(0,\infty)
    \end{align}
    for each $i$. Then $F[f_i](\alpha) = O(|\alpha|^{-2})$ as $\alpha \to \infty$ on $\bar{D} \setminus \{0\}$ for each $i$. 
\end{theorem}

\begin{proof}
    By assumption, $s^{1/2} h_i \in L^1(0,\infty)$ for each $i$, so linear combinations of $h_i$ and $f_i$ belong to $\F$. Thus $\LB f_i \in \F$. According to Lemma~\ref{lemma:L1function}, $F[\LB f_i](\alpha) + \sum_{j=1}^N A_{ij} F[f_j](\alpha)$ is continuous in $\alpha$ on $\bar{D} \setminus \{0\}$, and
    \begin{align}
        F[h_i](\alpha) = F[\LB f_i](\alpha) + \sum_{j=1}^N A_{ij} F[f_j](\alpha) = O(|\alpha|^{-1/2})
    \end{align}
    as $\alpha \to \infty$. 

    On the other hand, for $\alpha \in D$, using Lemma~\ref{lemma:IBP}, we have
    \begin{align}
        F[\LB f_i](\alpha) = \alpha^2 F[f_i](\alpha) - 2 f_i(0) \ .
    \end{align}
    Combining these two equations gives
    \begin{align}
        F[h_i](\alpha) + 2 f_i(0) = \sum_j (\alpha^2 \delta_{ij} + A_{ij}) F[f_j](\alpha)  \ .
    \end{align}
    Let $R > \| A \|$ be a constant. Then for $|\alpha| > R$, the matrix $\alpha^2 I + A$ is invertible, and we have 
    \begin{align}
        F[f_i](\alpha) = \sum_j (\alpha^2 I + A)^{-1}_{ij} (F[h_j](\alpha) + 2 f_j(0)) \ .
    \end{align}

    Since $F[h_j](\alpha) = O(|\alpha|^{-1/2})$ as $\alpha \to \infty$, we have $F[f_i](\alpha) = O(|\alpha|^{-2})$ as $\alpha \to \infty$. Moreover, by Theorem~\ref{theorem:continuity},
    \begin{align}
        F[f_i](\rmi \omega) = \lim_{\epsilon \to 0^+} F[f_i](\epsilon + \rmi \omega)
    \end{align}
    as generalized functions. Since $F[h_j](\alpha)$ is continuous in $\alpha$ on $\bar{D} \setminus \{0\}$ and $(\alpha^2 I + A)^{-1}$ is continuous in $\alpha$ on $\{z \in \mathbb{C} : |z| \geq R\}$, $F[f_i](\alpha)$ is continuous in $\alpha$ on $\bar{D} \cap \{z \in \mathbb{C} : |z| \geq R\}$. Therefore, $F[f_i](\alpha) = O(|\alpha|^{-2})$ as $\alpha \to \infty$ on $\bar{D} \setminus \{0\}$ for each $i$.  
\end{proof}

We now consider the integral $F[\h^i (\h^{\prime})^j s^{-k}] \eqqcolon F(i,j,k)$. 

We first treat the cases with $1 \leq i + j + 2 k < 4$. Under the assumption $j \geq k$, the possibilities are:
\begin{enumerate}
    \item[{\bf 1.}] $i = 0, j = 1, k = 1$, which corresponds to $F(0,1,1)$. 
    \item[{\bf 2.}] $i + j = 3, k = 0$, which corresponds to $F(3,0,0), F(2,1,0), F(1,2,0)$ and $F(0,3,0)$. 
    \item[{\bf 3.}] $i + j = 2, k = 0$, which corresponds to $F(2,0,0), F(1,1,0)$ and $F(0,2,0)$. 
    \item[{\bf 4.}] $i + j = 1, k = 0$, which corresponds to $F(1,0,0)$ and $F(0,1,0)$. 
\end{enumerate}

A useful observation is that $\Delta(\h) = O(s^{-1})$ and $\Delta^\prime(\h) = O(s^{-1/2})$ as $s \to \infty$, since $\h, \h^\prime = O(s^{-1/2})$. We will prove the following theorem: 
\begin{theorem} \label{theorem:leq3}
    For $1 \leq i + j + 2 k < 4$ and $j \geq k$, the integral $F(i,j,k) = O(|\alpha|^{-2})$ as $\alpha \to \infty$.
\end{theorem}

\begin{proof}
    Since $j \geq k$, the function $H_{ijk}$ is smooth on $[0,\infty)$. Hence $s^{1/2} H_{ijk} \in L^1_{\rm loc}(0,\infty)$ and $s^{1/2} \LB H_{ijk} = s^{1/2} H_{ijk}^{\prime\prime} + s^{-1/2} H_{ijk}^{\prime} \in L^1_{\rm loc}(0,\infty)$. Moreover, $\h, \h^{\prime} = O(s^{-1/2})$ implies $H_{ijk}(s) = O(s^{-(i+j+2k)/2})$ as $s \to \infty$. We now proceed case by case. 

    By direct computation, we have the following equations: 
    \begin{align}
        &~(\LB + M_h^2) H_{011} \in B_4 \ , \\ 
        &~(\LB + 3 M_h^2) H_{300} - 6 H_{120} \in B_4 \ , \\ 
        &~(\LB + 7 M_h^2) H_{120} - 2 M_h^4 H_{300} \in B_4 \ , \\ 
        &~(\LB + 7 M_h^2) H_{210} - 2 H_{030} \in B_4 \ , \\
        &~(\LB + 3 M_h^2) H_{030} - 6 M_h^4 H_{210} \in B_4 \ , \\ 
        &~(\LB + 2 M_h^2) H_{200} - 2 H_{020} + \Delta^{\prime\prime}(0) H_{300} \in B_4 \ , \\ 
        &~(\LB + 4 M_h^2) H_{110} + \frac{5}{2} \Delta^{\prime\prime}(0) H_{210} \in B_4 \ , \\ 
        &~(\LB + 2 M_h^2) H_{020} - 2 M_h^4 H_{200} + 2 M_h^2 \Delta^{\prime\prime}(0) H_{300} + 2 \Delta^{\prime\prime}(0) H_{120} \in B_4 \ , \\ 
        &~(\LB + M_h^2) H_{100} + \frac{1}{2} \Delta^{\prime\prime}(0) H_{200} + \frac{1}{6} \Delta^{(3)}(0) H_{300} \in B_4 \ , \\ 
        &~(\LB + M_h^2) H_{010} + \Delta^{\prime\prime}(0) H_{110} + \frac{1}{2} \Delta^{(3)}(0) H_{210} \in B_4 \ . 
    \end{align}

    Since all functions above lie in $L^1_{\rm loc}(0,\infty)$, Corollary~\ref{corollary:bound} implies that they actually belong to $L^1(0,\infty)$. Thus $H_{ijk}$ satisfies the conditions of Theorem~\ref{theorem:MultipleFunctions}. Therefore $F(i,j,k) = O(|\alpha|^{-2})$ as $\alpha \to \infty$ for all $(i,j,k)$ with $1 \leq i + j + 2 k < 4$ and $j \geq k$. 
\end{proof}

We can now prove the main theorem.
\begin{theorem} \label{theorem:main}
    The integral $F[\LB^n \Delta] = O(|\alpha|^{-2})$ as $\alpha \to \infty$.
\end{theorem}

\begin{proof}
    Since $\Delta$ is smooth in $\h$ and satisfies $\Delta(0) = \Delta^\prime(0) = 0$, we have 
    \begin{align}
        \Delta_4(\h) = \Delta(\h(s)) - \frac{1}{2} \Delta^{\prime\prime}(0) H_{200} - \frac{1}{6} \Delta^{(3)}(0) H_{300} \in A_4 \subset B_4
    \end{align}
    as $\h \to 0$. Then the function $U(\h) = \Delta_4(\h) / \h^4$ is smooth with respect to $\h$, and we have $\Delta_4(\h) = U(\h) \h^4 \in A_4$. 

    On the one hand, since $\LB^{n+1} \Delta(\h), \LB^{n+1} H_{200}, \LB^{n+1} H_{020}, \LB^{n+1} H_{300}, \LB^{n+1} H_{120}$ are smooth with respect to $s$ by Lemma~\ref{lemma:smoothness}, they are locally $L^1$-integrable. 

    On the other hand, $B_4$ is invariant under $\LB$, so we have 
    \begin{align}\label{eq:it}
        &~(\LB + 3 M_h^2) \LB^{n+1} H_{300} - 6 \LB^{n+1} H_{120} \in B_4 \ , \\ 
        &~(\LB + 7 M_h^2) \LB^{n+1} H_{120} - 2 M_h^4 \LB^{n+1} H_{300} \in B_4 \ , \\ 
        &~(\LB + 2 M_h^2) \LB^{n+1} H_{200} - 2 \LB^{n+1} H_{020} + \Delta^{\prime\prime}(0) \LB^{n+1} H_{300} \in B_4 \ , \\ 
        &~(\LB + 2 M_h^2) \LB^{n+1} H_{020} - 2 M_h^4 \LB^{n+1} H_{200} + 2 M_h^2 \Delta^{\prime\prime}(0) \LB^{n+1} H_{300} + 2 \Delta^{\prime\prime}(0) \LB^{n+1} H_{120} \in B_4 \ , \\ 
        &~\LB \LB^n \Delta(\h(s)) - \frac{1}{2} \Delta^{\prime\prime}(0) \LB^{n+1} H_{200} - \frac{1}{6} \Delta^{(3)}(0) \LB^{n+1} H_{300} \in B_4 \ . 
    \end{align}
    
    Therefore the functions above are all in $L^1(0,\infty)$, and we can apply Theorem~\ref{theorem:MultipleFunctions} to the set $\{f_i\}_{i=1}^5 = \{ \LB^n \Delta, \LB^{n+1} H_{200}, \LB^{n+1} H_{020}, \LB^{n+1} H_{300}, \LB^{n+1} H_{120} \}$ and obtain $F[\LB^n \Delta] = O(|\alpha|^{-2})$ as $\alpha \to \infty$ for any $n \in \mathbb{N}$. 
\end{proof}

Finally, we derive the asymptotic expansion of $I(\chi)$.
\begin{corollary} \label{corollary3}
    The integral $I(\chi)$ can be written as 
    \begin{align}
        I(\chi) = \sum_{n = 0}^{N-1} (-1)^n \frac{(\partial_s^2 + s^{-1} \partial_s)^n \Delta(\h(s))}{\chi^{n+1}} \Bigg|_{s=0} + O(\chi^{-N-1})
    \end{align}
\end{corollary}

\begin{proof}
    According to Eq.~\eqref{eq:it}, there exist constants $A_{n,i}$ such that 
    \begin{align}
        \LB^n \Delta(\h) - A_{n,1} H_{200} - A_{n,2} H_{020} - A_{n,3} H_{300} - A_{n,4} H_{120} \in B_4 \ . 
    \end{align}
    Therefore, for any $n \in \mathbb{N}$, the function $\LB^n \Delta(\h)$ belongs to ${\cal F}$. By Lemma~\ref{lemma:smoothness}, this function is also smooth. Thus we can apply Lemma~\ref{lemma:IBP} to $\LB^n \Delta(\h)$. Repeated integration by parts produces the first several terms in the asymptotic expansion of $I(\chi)$ as $\chi \to \infty$. By Theorem~\ref{theorem:main}, the remainder is $O(\chi^{-N-1})$ since $F[\LB^n \Delta] = O(|\alpha|^{-2})$ as $\alpha \to \infty$. This completes the proof. 
\end{proof}

%% file: appendixC.tex
\section{Singular part of \texorpdfstring{$\partial^2 V$}{}}\label{appendixC}

In this appendix, we present the details of the derivation of the singular part of $\partial^2 V$ in the large-boost limit, which is crucial for evaluating particle production in bubble collisions in $(3+1)$ dimensions. The singular part arises from derivatives of the step functions used to approximate the bubble walls in the large-boost limit, yielding
\begin{align}
    (\partial^2 V^\prime(\phi))_{\rm sing} =&~ \sum_{a \in \{0,1\}^n} V^\prime\bigg(\phi_{\rm c} + \sum_i a_i \phi_{{\rm d},i}\bigg) \sum_{j \neq k} (\partial_\mu \chi(B_j^{a_j})) (\partial^\mu \chi(B_k^{a_k})) \prod_{i \neq j,k} \chi(B_i^{a_i}) \nnn 
    &+ \sum_{a \in \{0,1\}^n} V^\prime\bigg(\phi_{\rm c} + \sum_i a_i \phi_{{\rm d},i}\bigg) \sum_{j} \partial^2 \chi(B_j^{a_j}) \prod_{i \neq j} \chi(B_i^{a_i}) \nnn
    &+ \sum_{a \in \{0,1\}^n} 2 \partial_\mu V^\prime\bigg(\phi_{\rm c} + \sum_i a_i \phi_{{\rm d},i}\bigg) \sum_{j} \partial^\mu \chi(B_j^{a_j}) \prod_{i \neq j} \chi(B_i^{a_i}) \ .
\end{align}

The first term is supported on the intersection of two bubble walls, while the second and third terms are supported on the bubble walls themselves. Accordingly, the first term corresponds to the contribution from bubble collisions, whereas the second and third terms correspond to the contribution from bubble expansion. We therefore write
\begin{align}
    (\partial^2 V^\prime(\phi))_{\rm coll} =&~ \sum_{a \in \{0,1\}^n} V^\prime\bigg(\phi_{\rm c} + \sum_i a_i \phi_{{\rm d},i}\bigg) \sum_{j \neq k} (\partial_\mu \chi(B_j^{a_j})) (\partial^\mu \chi(B_k^{a_k})) \prod_{i \neq j,k} \chi(B_i^{a_i}) \ , 
\end{align}
and 
\begin{align}
    (\partial^2 V^\prime(\phi))_{\rm exp} =&~ \sum_{a \in \{0,1\}^n} V^\prime\bigg(\phi_{\rm c} + \sum_i a_i \phi_{{\rm d},i}\bigg) \sum_{j} \partial^2 \chi(B_j^{a_j}) \prod_{i \neq j} \chi(B_i^{a_i}) \nnn
    &+ \sum_{a \in \{0,1\}^n} 2 \partial_\mu V^\prime\bigg(\phi_{\rm c} + \sum_i a_i \phi_{{\rm d},i}\bigg) \sum_{j} \partial^\mu \chi(B_j^{a_j}) \prod_{i \neq j} \chi(B_i^{a_i}) \ . 
\end{align}

Since $\partial^\mu \chi(B_j^{a_j}) = (-1)^{a_j + 1} \partial^\mu \chi(B_j)$, we can rewrite these terms as
\begin{align}
    (\partial^2 V^\prime(\phi))_{\rm coll} =&~ \sum_{j \neq k} (\partial_\mu \chi(B_j)) (\partial^\mu \chi(B_k)) \sum_{a \in \{0,1\}^n} V^\prime\bigg(\phi_{\rm c} + \sum_i a_i \phi_{{\rm d},i}\bigg) (-1)^{a_j + a_k}  \prod_{i \neq j,k} \chi(B_i^{a_i}) \ , 
\end{align}
and 
\begin{align}
    &~(\partial^2 V^\prime(\phi))_{\rm exp} \nnn
    =&~ \sum_{j} \partial^2 \chi(B_j) \sum_{a_{i (\neq j)} } \bigg[V^\prime\bigg(\phi_{\rm c} + \sum_{i \neq j} a_i \phi_{{\rm d},i} + \phi_{{\rm d},j}\bigg) - V^\prime\bigg(\phi_{\rm c} + \sum_{i \neq j} a_i \phi_{{\rm d},i}\bigg)\bigg] \prod_{i \neq j} \chi(B_i^{a_i}) \nnn
    &+ \sum_{j} 2 \partial^\mu \chi(B_j) \sum_{a_{i (\neq j)} } \partial_\mu \bigg[V^\prime\bigg(\phi_{\rm c} + \sum_{i \neq j} a_i \phi_{{\rm d},i} + \phi_{{\rm d},j}\bigg) - V^\prime\bigg(\phi_{\rm c} + \sum_{i \neq j} a_i \phi_{{\rm d},i}\bigg)\bigg]  \prod_{i \neq j} \chi(B_i^{a_i}) \ . 
\end{align}

\subsection{The collision part}

For the collision part, we have 
\begin{align}
    &~(\partial^2 V^\prime(\phi))_{\rm coll} = \sum_{a \in \{0,1\}^n} V^\prime\bigg(\phi_{\rm c} + \sum_i a_i \phi_{{\rm d},i}\bigg) \sum_{j \neq k} (\partial_\mu \chi(B_j^{a_j})) (\partial^\mu \chi(B_k^{a_k})) \prod_{i \neq j,k} \chi(B_i^{a_i}) \nnn
    =&~ \sum_{j \neq k} (\partial_\mu \chi(B_j)) (\partial^\mu \chi(B_k)) \sum_{a \in \{0,1\}^n} (-1)^{a_j + a_k} V^\prime\bigg(\phi_{\rm c} + \sum_i a_i \phi_{{\rm d},i}\bigg) \prod_{i \neq j,k} \chi(B_i^{a_i}) \nnn 
    \eqqcolon&~ \sum_{j \neq k} (\partial_\mu \chi(B_j)) (\partial^\mu \chi(B_k)) \bar{V}_{j,k}^\prime(\phi) \ , 
\end{align}
where 
\begin{align}
    \bar{V}_{j,k}^\prime(\phi) \coloneqq&~ \sum_{a \in \{0,1\}^n} (-1)^{a_j + a_k} V^\prime\bigg(\phi_{\rm c} + \sum_i a_i \phi_{{\rm d},i}\bigg) \prod_{i \neq j,k} \chi(B_i^{a_i}) \nnn 
    =&~ \sum_{a_{i (\neq j,k)} =0,1} \Big[V^\prime \Big(\phi^{0,0}_{j,k}\Big) - V^\prime\Big(\phi^{1,0}_{j,k}\Big) - V^\prime\Big(\phi^{0,1}_{j,k}\Big) + V^\prime\Big(\phi^{1,1}_{j,k}\Big)\Big]  \prod_{i \neq j,k} \chi(B_i^{a_i}) \ , \\
    \phi^{a_j,a_k}_{j,k} \coloneqq&~ \phi_{\rm c} + \sum_{i \neq j,k} a_i \phi_{{\rm d},i} + a_j \phi_{{\rm d},j} + a_k \phi_{{\rm d},k} \ .
\end{align}

\input{fig_intersection}

The quantity $(\partial_\mu \chi(B_j)) (\partial^\mu \chi(B_k)) \bar{V}_{j,k}^\prime(\phi)$ depends only on the field configuration at the intersection of the boundaries $\partial B_j$ and $\partial B_k$. To understand the physical meaning of $\bar{V}_{j,k}^\prime(\phi)$, we choose four points $x_{++}, x_{+-}, x_{-+}, x_{--}$ in the four domains $B_j \cap B_k$, $B_j \cap B_k^c$, $B_j^c \cap B_k$, and $B_j^c \cap B_k^c$, respectively, and take them to be close to a point $x \in \partial B_j \cap \partial B_k$ (see Figure \ref{fig:intersection}). Then
\begin{align}
    \bar{V}_{j,k}^\prime(\phi) = \sum_{a,b \in \{+,-\}} a b V^\prime(\phi(x_{ab})) \ ,
\end{align}
so it can be computed directly from the field configuration.

\subsection{Handling the singularity in the expansion part}

For the expansion part, we use the fact that $2 (\partial_\mu \phi_{{\rm d},i}) (\partial^\mu \chi(B_i)) + \phi_{{\rm d},i} \partial^2 \chi(B_i)$ has no singular part. Let $\Phi_i \coloneqq 2 (\partial_\mu \phi_{{\rm d},i}) (\partial^\mu \chi(B_i)) + \phi_{{\rm d},i} \partial^2 \chi(B_i)$. Then $\Phi_i$ is a distribution supported on the bubble wall $\partial B_i$. 

For a test function $f(x)$, we have
\begin{align}
    &~\langle \Phi_i, f \rangle = \int \d^4 x \ \Phi_i(x) f(x) = \int \d^4 x \ \Big[2 (\partial_\mu \phi_{{\rm d},i}) (\partial^\mu \chi(B_i)) + \phi_{{\rm d},i} \partial^2 \chi(B_i)\Big] f(x) \nnn 
    =&~ \int \d^4 x \ \chi(B_i)(x) (\partial^2 (\phi_{{\rm d},i} f) - 2 \partial_\mu (\phi_{{\rm d},i} \partial^\mu f)) \nnn
    =&~ \int_{\partial B_i} \d S^\mu \ \Big[f \partial_\mu \phi_{{\rm d},i} - \phi_{{\rm d},i} \partial_\mu f\Big] \ . 
\end{align}
Since $\partial B_i$ has measure zero, $\Phi_i$ is regular if and only if it vanishes as a distribution, which is equivalent to $\langle \Phi_i, f \rangle = 0$ for any test function $f$. Thus,
\begin{align}
    0 = \int_{\partial B_i} \d S^\mu \ \Big[f \partial_\mu \phi_{{\rm d},i} - \phi_{{\rm d},i} \partial_\mu f\Big] \ . 
\end{align}

For any $f$ with $f|_{\partial B_i} = 0$, the first term vanishes automatically, so we have
\begin{align}
    \int_{\partial B_i} \d S^\mu \ \phi_{{\rm d},i} \partial_\mu f = 0 \ . 
\end{align}
Since $f$ vanishes on the bubble wall, its tangential derivatives vanish there, so $\partial^\mu f$ is proportional to the normal vector $n_i^\mu$ of the bubble wall. Moreover, $\d S^\mu$ is also proportional to $n_i^\mu$. By the arbitrariness of $f$, the equation above implies that $n_i^\mu n_{i,\mu} = 0$ at every point where $\phi_{{\rm d},i}$ is nonzero on the bubble wall. On the other hand, if $\phi_{{\rm d},i} = 0$ at some point $x_0$, then the field configuration is continuous at $x_0$, so there is no singularity there and hence no constraint on the normal vector at $x_0$. 

Let ${\cal N}$ be the set of points on the bubble wall where $n_i^\mu$ is light-like. Since $\{ x \in \partial B_i : \phi_{{\rm d},i}(x) \neq 0 \} \subset {\cal N}$, the hypersurface ${\cal N}$ should be a closed subset of the bubble wall with nonempty interior. The integral above can then be rewritten as
\begin{align}
    \int_{\partial B_i} \d S^\mu \ f \partial_\mu \phi_{{\rm d},i} - \int_{\cal N} \d S^\mu \ \phi_{{\rm d},i} \partial_\mu f = 0 \ . 
\end{align}
If $f$ is supported only on a small subset of $(\partial B_i) \setminus {\cal N}$, then the second term vanishes, and we obtain
\begin{align}
    \int_{(\partial B_i) \setminus {\cal N}} \d S^\mu \ f \partial_\mu \phi_{{\rm d},i} = 0
\end{align}
for any $f$ with support on $\partial B_i \setminus {\cal N}$. Since $f$ is arbitrary, we have $n_i^\mu \partial_\mu \phi_{{\rm d},i} = 0$ on $\partial B_i \setminus {\cal N}$. Since $\phi_{{\rm d},i} = 0$ on $\partial B_i \setminus {\cal N}$, and by definition of ${\cal N}$, we have $n_i^\mu n_{i,\mu} \neq 0$ on $\partial B_i \setminus {\cal N}$. Thus both the field value $\phi_{{\rm d},i}$ and the derivative $\partial_\mu \phi_{{\rm d},i}$ vanish on $\partial B_i \setminus {\cal N}$. 

Now consider the integral over ${\cal N}$. Since $n_i^\mu$ is both the normal vector and a tangent vector to the bubble wall on ${\cal N}$, the value of this integral is intrinsically determined by the field configuration on ${\cal N}$. Let $\epsilon^{\cal N}$ denote the volume $3$-form on the bubble wall; then $\d S^\mu = n_i^\mu \epsilon^{\cal N}$. Using integration by parts, we obtain
\begin{align}
    &~\int_{\cal N} \d S^\mu \ \Big[f \partial_\mu \phi_{{\rm d},i} - \phi_{{\rm d},i} \partial_\mu f\Big] 
    = \int_{\cal N} \epsilon^{\cal N} \ (2 {\cal L}_{n_i} \phi_{{\rm d},i} + {\rm div}(n_i) \phi_{{\rm d},i}) f
    - \int_{\partial {\cal N}} (\iota_n \epsilon^{\cal N}) \ \phi_{{\rm d},i} f \ ,  
\end{align}
where $\iota_n \epsilon^{\cal N}$ is the interior product of $n_i$ and $\epsilon^{\cal N}$, which is a $2$-form on $\partial {\cal N}$. The boundary term vanishes since $\partial {\cal N}$ is a closed subset of the zero set of $\phi_{{\rm d},i}$. Since $f$ is arbitrary, we have $2 {\cal L}_{n_i} \phi_{{\rm d},i} + {\rm div}(n_i) \phi_{{\rm d},i} = 0$ on ${\cal N}$, which is equivalent to ${\cal L}_{n_i} (\phi_{{\rm d},i}^2 \epsilon^{\cal N}) = 0$ on ${\cal N}$. 

\subsection{The expansion part}

In the previous subsection, we showed that the singularity in the expansion part can arise only from points on the bubble wall where the normal vector is light-like. Moreover, $\Phi_i$ vanishes as a distribution. Thus we can rewrite the expansion term as
\begin{align}
    &~(\partial^2 V^\prime(\phi))_{\rm exp} \nnn
    =&~ \sum_{j} \partial^2 \chi(B_j) \sum_{a_{i (\neq j)}} \bigg[V^\prime\bigg(\phi_{\rm c} + \sum_{i \neq j} a_i \phi_{{\rm d},i} + \phi_{{\rm d},j}\bigg) - V^\prime\bigg(\phi_{\rm c} + \sum_{i \neq j} a_i \phi_{{\rm d},i}\bigg)\bigg] \prod_{i \neq j} \chi(B_i^{a_i}) \nnn
    &+ \sum_{j} 2 \partial^\mu \chi(B_j) \sum_{a_{i (\neq j)}} \partial_\mu \bigg[V^\prime\bigg(\phi_{\rm c} + \sum_{i \neq j} a_i \phi_{{\rm d},i} + \phi_{{\rm d},j}\bigg) - V^\prime\bigg(\phi_{\rm c} + \sum_{i \neq j} a_i \phi_{{\rm d},i}\bigg)\bigg]  \prod_{i \neq j} \chi(B_i^{a_i}) \nnn
    &- \sum_{j} [\phi_{{\rm d},j} \partial^2 \chi(B_j) + 2 \partial_\mu \phi_{{\rm d},j} \partial^\mu \chi(B_j)] \sum_{a_{i (\neq j)}} V^{\prime\prime}\bigg(\phi_{\rm c} + \sum_{i \neq j} a_i \phi_{{\rm d},i}\bigg) \prod_{i \neq j} \chi(B_i^{a_i}) \nnn 
    \eqqcolon&~ \sum_{j} \partial^2 \chi(B_j) F_j(\phi) 
    + \sum_{j} 2 \partial_\mu \chi(B_j) G^\mu_j(\phi) \ , 
\end{align}
where 
\begin{align}
    F_j(\phi) \coloneqq&~ \sum_{a_{i (\neq j)}} [
        V^\prime(\phi^1_j) 
        - V^\prime(\phi^0_j)
        - V^{\prime\prime}(\phi^0_j) \phi_{{\rm d},j}
    ] \prod_{i \neq j} \chi(B_i^{a_i}) \ , \\ 
    G^\mu_j(\phi) \coloneqq&~ \sum_{a_{i (\neq j)}} (\partial^\mu \phi^1_j) [
        V^{\prime\prime}(\phi^1_j) 
        - V^{\prime\prime}(\phi^0_j)
    ] \prod_{i \neq j} \chi(B_i^{a_i}) \ , \\
    \phi^a_j \coloneqq&~ \phi_{\rm c} + \sum_{i \neq j} a_i \phi_{{\rm d},i} + a \phi_{{\rm d},j} \ .
\end{align}

Using the same method as in the previous subsection, we find that the singularity in the expansion part can only arise from points on the bubble wall where the normal vector is light-like. Thus we can rewrite the expansion term as
\begin{align}
    (\partial^2 V^\prime(\phi))_{\rm exp} = \sum_j (2 n_{j,\mu} G^\mu_j(\phi) + {\rm div}(n_j) F_j(\phi)) \delta(\partial B_j) \ ,
\end{align}

If the potential is quadratic, we immediately have $\bar{V}_{j,k}^\prime(\phi) = F_j(\phi) = G^\mu_j(\phi) = 0$, so both the collision and expansion terms vanish. For a general potential, these terms need not vanish. However, if we assume that the walls propagate in vacuum, then $(\partial^\mu \phi_j^1) \partial_\mu \chi(B_j) = 0$ on the bubble walls, and hence $G^\mu_j(\phi) = 0$. The term $F_j(\phi)$ is still generically nonzero, so bubble expansion gives a nonzero contribution to particle production. Nevertheless, the divergence of the normal vector is typically of order $1/R$, where $R$ is the bubble radius, so the contribution from bubble expansion is suppressed by $1/R$ relative to the contribution from bubble collisions. Therefore, in the large-boost limit, the contribution from bubble expansion can be neglected.

%% file: fig_intersection.tex
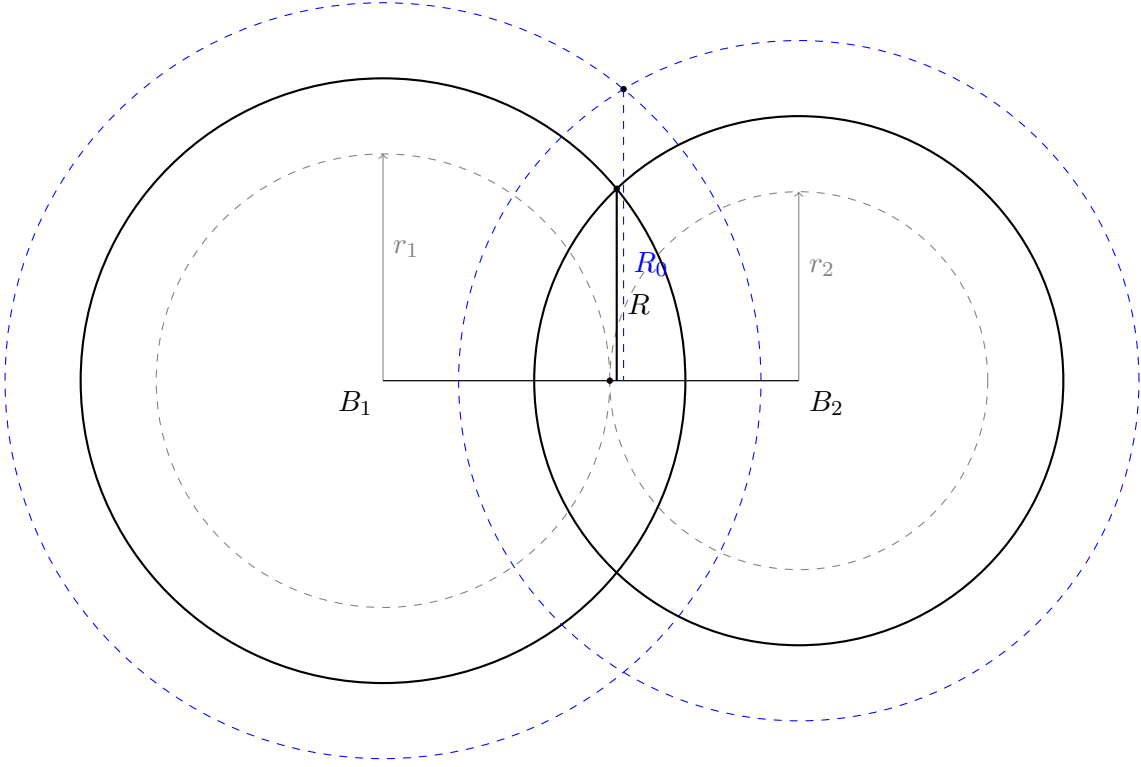
\begin{figure}[t]     
    \centering
    \begin{tikzpicture}

        \def\rL{3.0}
        \def\mL{4.0}
        \def\RL{5.0}
        \def\rR{2.5}
        \def\mR{3.5}
        \def\RR{4.5}

        \pgfmathsetmacro{\dist}{\rL+\rR}
        \coordinate     (C1)  at  (0,0);
        \coordinate     (C2)  at  (\dist,0);

        \draw[gray,dashed]  (C1)    circle (\rL);
        \draw[thick]        (C1)    circle (\mL);
        \draw[blue,dashed]  (C1)    circle (\RL);

        \draw[gray,dashed]  (C2)    circle (\rR);
        \draw[thick]        (C2)    circle (\mR);
        \draw[blue,dashed]  (C2)    circle (\RR);

        \draw[gray]     (C1) -- (C2);

        \coordinate     (T) at (\rL,0); 
        \fill           (T) circle (1.2pt);

        \coordinate     (U1) at (0,\rL);
        \coordinate     (U2) at (\dist,\rR);
        \draw[gray,->]     (C1) -- (U1)     node[midway, above right] {$r_1$};
        \draw[gray,->]     (C2) -- (U2)     node[midway, above right] {$r_2$};

        \path[name path=leftm]  (C1) circle (\mL);
        \path[name path=rightm] (C2) circle (\mR);
        \path[name path=leftR]  (C1) circle (\RL);
        \path[name path=rightR] (C2) circle (\RR);
        \path[name intersections={of=leftm and rightm, by={I1,I2}}];
        \path[name intersections={of=leftR and rightR, by={I3,I4}}];
        \fill (I1) circle (1.2pt);
        \fill (I3) circle (1.2pt);

        \draw[black]            (C1) -- (T);
        \draw[black]            (C2) -- (T);
        \draw[thick]            (I1) -- (I1 |- T)         node[pos=0.6, right] {$R$};
        \draw[blue,dashed]      (I3) -- (I3 |- T)         node[pos=0.6, right] {$R_0$};

        \node[below left]   at (C1) {$B_1$};
        \node[below right]  at (C2) {$B_2$};
    \end{tikzpicture}
    \caption{The collision of two bubbles $B_1$ and $B_2$. The gray dashed circles represent the bubble walls at the beginning of the collision, while the blue dashed circles represent the bubbles at the end. The radii of the two bubbles at the moment they collide are denoted as $r_1$ and $r_2$, respectively. The intersection of the two walls is a circle (two separate points in this figure since its a projection to the plane) with radius $R$, which will finally be $R_0$.}
    \label{fig:collision}
\end{figure}